\PassOptionsToPackage{numbers, sort&compress}{natbib}
\documentclass[acmsmall,screen]{acmart}

\usepackage[utf8]{inputenc} 
\usepackage[T1]{fontenc}    
\usepackage{hyperref}       
\usepackage{url}            
\usepackage{booktabs}       
\usepackage{amsfonts}       
\usepackage{nicefrac}       
\usepackage{microtype}      

\usepackage{bm}
\usepackage{amsmath}
\usepackage{graphicx}
\usepackage[linesnumbered,ruled,vlined]{algorithm2e}
\usepackage[caption=false,font=footnotesize]{subfig}
\graphicspath{{figures/}}
\newcommand{\eg}{{e.g.,}\xspace}
\newcommand{\ie}{{i.e.,}\xspace}
\newcommand{\etal}{{et al.\@}\xspace}
\newcommand{\spara}[1]{\vspace{1mm}\noindent\textbf{#1.}}
\newcommand{\OPT}{\operatorname{OPT}}
\newcommand{\e}{{\ensuremath{\mathrm{e}}}}

\newcommand{\MG}{\textsc{MGreedy}\xspace}
\newcommand{\Sm}{\ensuremath{S_{\mathrm{m}}}}
\newcommand{\Sg}{\ensuremath{S_{\mathrm{g}}}}
\newcommand{\Ss}{\ensuremath{S^\ast}}
\newcommand{\As}{\ensuremath{A^\ast}}

\AtBeginDocument{%
	\providecommand\BibTeX{{%
			\normalfont B\kern-0.5em{\scshape i\kern-0.25em b}\kern-0.8em\TeX}}}

\setcopyright{rightsretained}
\acmJournal{POMACS}
\acmYear{2021} \acmVolume{5} \acmNumber{1} 
\acmArticle{8} \acmMonth{3} \acmPrice{15.00}
\acmDOI{10.1145/3447386}

\received{October 2020} 
\received[revised]{December 2020}
\received[accepted]{January 2021}

\begin{document}
	
\title{Revisiting Modified Greedy Algorithm for Monotone Submodular Maximization with a Knapsack Constraint}
\titlenote{The paper will appear in 2021 ACM SIGMETRICS conference (SIGMETRICS '21), June 14--18, 2021, Beijing, China.}


\author{Jing Tang}
\orcid{0000-0002-0785-707X}
\email{isejtang@nus.edu.sg}
\affiliation{%
	\department{Department of Industrial Systems Engineering and Management}
	\institution{National University of Singapore}
	\country{Singapore}
}

\author{Xueyan Tang}
\email{asxytang@ntu.edu.sg}
\affiliation{%
	\department{School of Computer Science and Engineering}
	\institution{Nanyang Technological University}
	\country{Singapore}
}

\author{Andrew Lim}
\email{isealim@nus.edu.sg}
\affiliation{%
	\department{Department of Industrial Systems Engineering and Management}
	\institution{National University of Singapore}
	\country{Singapore}
}

\author{Kai Han}
\email{hankai@ustc.edu.cn}
\affiliation{%
	\department{School of Computer Science and Technology}
	\institution{University of Science and Technology of China}
	\country{China}
}

\author{Chongshou Li}
\email{iselc@nus.edu.sg}
\affiliation{%
	\department{Department of Industrial Systems Engineering and Management}
	\institution{National University of Singapore}
	\country{Singapore}
}

\author{Junsong Yuan}
\email{jsyuan@buffalo.edu}
\affiliation{%
	\department{Department of Computer Science and Engineering}
	\institution{State University of New York at Buffalo}
	\country{USA}
}

\begin{abstract}
	Monotone submodular maximization with a knapsack constraint is NP-hard. Various approximation algorithms have been devised to address this optimization problem. In this paper, we revisit the widely known modified greedy algorithm. First, we show that this algorithm can achieve an approximation factor of $0.405$, which significantly improves the known factors of $0.357$ given by \citet{Wolsey_MGreedy_1982} and $(1-1/\e)/2\approx 0.316$ given by \citet{Khuller_BMCP_1999}. More importantly, our analysis closes a gap in Khuller et al.'s proof for the extensively mentioned approximation factor of $(1-1/\sqrt{\e})\approx 0.393$ in the literature to clarify a long-standing misconception on this issue. Second, we enhance the modified greedy algorithm to derive a data-dependent upper bound on the optimum. We empirically demonstrate the tightness of our upper bound with a real-world application. The bound enables us to obtain a data-dependent ratio typically much higher than $0.405$ between the solution value of the modified greedy algorithm and the optimum. It can also be used to significantly improve the efficiency of algorithms such as branch and bound. 
\end{abstract}


	
\maketitle

\begin{sloppy}
	\section{Introduction}\label{sec:introduction}
A set function $f\colon 2^V\rightarrow \mathbb{R}$ is \textit{submodular} \cite{Nemhauser_submodular_1978} if for all $S,T\subseteq V$, it holds that \[f(S)+f(T)\geq f(S\cup T)+f(S\cap T).\] Alternatively, defining \[f(v\mid S):= f(S\cup\{v\})-f(S)\] as the marginal gain of adding an element $v\in V$ to a set $S\subseteq V$, an equivalent definition of a submodular set function $f$ is that for all $S \subseteq T$ and $v\in V\setminus T$, \[f(v\mid S)\geq f(v\mid T).\] The latter form of definition describes the concept of \textit{diminishing return} in economics.
The function $f$ is \textit{monotone} nondecreasing if and only if $f(S)\leq f(T)$ for all $S\subseteq T$ (or equivalently $f(v\mid S)\geq 0$).


Many well known combinatorial optimization problems are essentially monotone submodular maximization, such as maximum coverage~\cite{Feige_setcover_1998,Khuller_BMCP_1999}, maximum facility location~\cite{Cornuejols_FLP_1977,Ageev_UFLP_1999}, and maximum entropy sampling~\cite{Shewry_entropy_1987,Ko_entropy_1995}. In addition, a growing number of problems in real-world applications of artificial intelligence and machine learning are also shown to be monotone submodular maximization. These problems include sensor placement~\cite{Krause_waterSensor_2008,Krause_sensor_2008}, feature selection~\cite{Krause_feature_2005,Yu_feature_2016}, viral marketing~\cite{Kempe_maxInfluence_2003,Kempe_influence_2005}, image segmentation~\cite{Boykov_segmentation_2001,Delong_segmentation_2012,Jegelka_segmentation_2011}, document summarization~\cite{Lin_document_2010,Lin_document_2011}, data subset selection \cite{Krause_information_2007,Wei_dataselection_2015}, etc.


In this paper, we study monotone submodular maximization with a knapsack constraint, which is defined as follows:
\begin{equation}\label{eqn:problem}
	\max_{S\subseteq V} f(S) \text{ s.t.\@ } c(S)\leq b,
\end{equation}
where $f$ is a monotone nondecreasing submodular set function\footnote{We assume that function $f$ is normalized, \ie~$f(\emptyset)=0$, and is given via a {value oracle}.}, $c(S)=\sum_{v\in S}c(v)$ and $c(v)$ represents the cost of element $v$. Without loss of generality, we may assume that the cost of each element does not exceed $b$, since elements with cost greater than $b$ do not belong to any feasible solution. 
This optimization problem has already found great utility in the aforementioned applications.

Since this optimization problem is NP-hard in general, various approximation algorithms have been proposed. For a special cardinality constraint where the costs of all elements are identical, \ie~$c(v)=1$ for every element $v\in V$, \citet{Nemhauser_submodular_1978}~proposed a simple hill-climbing greedy algorithm that can provide $(1-1/\e)$-approximation (we say that an algorithm provides $\alpha$-approximation, where $\alpha\leq 1$, if it always obtains a solution of value at least $\alpha$ times the value of an optimal solution). However, for the general knapsack constraint, the approximation factor of such a greedy algorithm is unbounded. \citet{Wolsey_MGreedy_1982} found that slightly modifying the original greedy algorithm can provide an approximation ratio of $(1-1/\e^\beta)\approx 0.357$, where $\beta$ is the unique root of the equation $\e^x=2-x$. \citet{Khuller_BMCP_1999} studied the budgeted maximum coverage problem (a special case of monotone submodular maximization with a knapsack constraint where the function value is always an integer), and derived two approximation factors, \ie $(1-1/{\e})/2 \approx 0.316$ and $(1-1/\sqrt{\e})\approx 0.393$, for the modified greedy algorithm developed by \citet{Wolsey_MGreedy_1982}. We note that the factor of $(1-1/\sqrt{\e})$ is extensively mentioned in the literature, but unfortunately their proof was flawed as pointed out by \citet{Zhang_billboard_2018}. It becomes an open question whether the modified greedy algorithm can achieve an approximation ratio at least $(1-1/\sqrt{\e})$. 

\citet{Khuller_BMCP_1999} also developed a partial enumeration greedy algorithm that improves the approximation factor to $(1-1/\e)$, which was later shown to be applicable to the general problem \eqref{eqn:problem}~\cite{Sviridenko_maxSub_2004}. However, this algorithm requires $O(n^5)$ (where $n=\lvert V\rvert$ is the total number of elements in the ground set $V$) function value computations, which is not scalable. We focus on the scalable modified greedy algorithm of $O(n^2)$ \cite{Khuller_BMCP_1999} and conduct a comprehensive analysis on its worst-case approximation guarantee. Based on the monotonicity and submodularity, we derive several relations governing the solution value and the optimum. Leveraging these relations, we establish an approximation ratio of $0.405$, which significantly improves the factors of $(1-1/\e^\beta)\approx 0.357$ given by \citet{Wolsey_MGreedy_1982} and $(1-1/{\e})/2\approx 0.316$ given by \citet{Khuller_BMCP_1999}. More importantly, our analysis fills a critical gap in the proof for the factor of $(1-1/\sqrt{\e})\approx 0.393$ given by \citet{Khuller_BMCP_1999} to clarify a long-standing misconception in the literature.



In addition, we enhance the modified greedy algorithm to derive a data-dependent upper bound on the optimum. We empirically demonstrate the tightness of our bound with a real-world application of viral marketing in social networks. The bound enables us to obtain a data-dependent ratio typically much higher than $0.405$ between the solution value of the modified greedy algorithm and the optimum. It can also be used to significantly improve the efficiency of algorithms such as branch and bound as shown by our experimental evaluations.

	\section{Modified Greedy Algorithm and Approximation Guarantees}
For the unit cost version of the optimization problem defined in \eqref{eqn:problem}, a simple greedy algorithm that chooses the element with the largest marginal gain in each iteration can achieve an approximation factor of $(1-1/\e)$ \cite{Nemhauser_submodular_1978}. Inspired by this elegant algorithm, for the general cost version, it is natural to apply a similar greedy algorithm according to cost-effectiveness. That is, we pick in each iteration the element that maximizes the ratio $\frac{f(v\mid \Sg)}{c(v)}$ based on the selected element set $\Sg$. Unfortunately, this simple greedy algorithm  has an unbounded approximation factor. Consider, for example, two elements $u$ and $v$ with $f(\{u\})=1$, $f(\{v\})=2\varepsilon$, $c(u)=1$ and $c(v)=\varepsilon$, where $\varepsilon$ is a small positive number. When $b=1$, the optimal solution is $\{u\}$ while the greedy heuristic picks $\{v\}$. The approximation factor for this instance is $2\varepsilon$, and is therefore unbounded.


\begin{algorithm}[!t]
	\caption{\MG}\label{alg:MGreedy}
	initialize $\Sg\gets\emptyset$, $V'\gets V$\;\label{alg:MGreedy_filter}
	\While{$V'\neq \emptyset$}{
		find $u\gets\arg\max_{v\in V'}\big\{\frac{f(v\mid \Sg)}{c(v)}\big\}$\;\label{alg:MGreedy_argmax}
		\If{$c(S)+c(u)\leq b$\label{alg:MGreedy_constraint}}{$\Sg\gets \Sg\cup\{u\}$\;\label{alg:MGreedy_update}}
		update the search space $V'\leftarrow V'\setminus\{u\}$\;\label{alg:MGreedy_reducespace}
	}
	$v^\ast\gets \arg\max_{v\in V} f(v)$\;\label{alg:MGreedy_solution1}
	$\Sm\gets\arg\max_{S\in\{\{v^*\},\Sg\}}f(S)$\;\label{alg:MGreedy_solution}
	\Return{$\Sm$}\;\label{alg:MGreedy_return}
\end{algorithm}

Interestingly, a small modification to the greedy algorithm, referred to as \MG (Algorithm~\ref{alg:MGreedy}), achieves a constant approximation factor \cite{Khuller_BMCP_1999,Wolsey_MGreedy_1982}. Specifically, in addition to $\Sg$ obtained from the greedy heuristic (Lines~\ref{alg:MGreedy_filter}--\ref{alg:MGreedy_reducespace}), the algorithm also finds an element $v^\ast$ that maximizes $f(\{v\})$ (Line~\ref{alg:MGreedy_solution1}), and then chooses the better one between $\Sg$ and $\{v^\ast\}$ (Lines \ref{alg:MGreedy_solution} and \ref{alg:MGreedy_return}). 
\citet{Wolsey_MGreedy_1982} showed that \MG achieves an approximation factor of $0.357$. Later, \citet{Khuller_BMCP_1999} gave two approximation factors of $(1-1/\e)/2$ and $(1-1/\sqrt{\e})$ that can be achieved by \MG for the budgeted maximum coverage problem, but unfortunately their proof for the factor of $(1-1/\sqrt{\e})$ was flawed as pointed out by \citet{Zhang_billboard_2018}. In Section~\ref{subsec:correct-analysis}, we provide an explanation of the problem in Khuller \etal's proof and show a correct proof for the factor of $(1-1/\sqrt{\e})$.
A main result of this paper is to establish an improved approximation factor of $0.405$ for \MG. 
\begin{theorem}\label{thm:approx-ratio}
	Let $\alpha^\bot$ be the root of 
	\begin{equation}\label{eqn:alpha^bot}
		(1-\alpha^\bot)\ln(1-\alpha^\bot)+(2-1/\e)(1-2\alpha^\bot)=0
	\end{equation}
	satisfying $\alpha^\bot>0.405$. The \MG algorithm achieves an approximation factor of $\alpha^\bot$.
\end{theorem}
To our knowledge, this is the first work giving a constant factor achieved by \MG that is even larger than $(1-1/\sqrt{\e})\approx 0.393$, which not only significantly improves the known factors of $(1-1/\e^\beta)\approx 0.357$ given by \citet{Wolsey_MGreedy_1982} and $(1-1/\e)/2\approx 0.316$ given by \citet{Khuller_BMCP_1999} but also clarifies a long-term misunderstanding regarding the factor of $(1-1/\sqrt{\e})$ in the literature.

	\section{Proof of Theorem~\ref{thm:approx-ratio}}\label{sec:proof}
The key idea of our proof is that we derive several relations governing the solution value and the optimum by carefully characterizing the properties of \MG, and utilize these relations to construct an optimization problem whose optimum is a lower bound on the approximation factor of \MG. Then, it remains to show that the optimum of our newly constructed optimization problem is no less than $0.405$. Our analysis procedure can be used as a general approach for analyzing the approximation guarantees of algorithms. In the following, we first introduce some useful notations and definitions.


\subsection{Notations and Definitions}
To characterize \MG, we denote $u_i$ as the $i$-th element added to the greedy solution $\Sg$ by the greedy heuristic and $S_i:=\{u_1,u_2,\dotsc,u_i\}$ as the first $i$ added elements for any $0\leq i\leq \lvert \Sg\rvert$. Corresponding to the added elements, we denote $A_{i}$ as the element set abandoned due to budget violation until $u_{i}$ is selected for consideration. That is, after the round that $u_i$ is added to the greedy solution, the remaining element set becomes $V^\prime$ and all elements in the set $V\setminus V^\prime$ are either added or abandoned, \ie~$A_{i}=(V\setminus V^\prime)\setminus S_i$, where $V$ is the ground element set.
Based on the sequence $\Sg=\{u_1,u_2,\dotsc\}$ of greedy solution, given any $x\leq c(\Sg)$, we define $j$ as the index such that $c(S_{j})< x \leq c(S_{j+1})$ and further define a continuous extension $F(x)$ of the set function $f(\cdot)$ as   
\begin{equation}\label{eqn:F(x)}
F(x):=f(S_j)+f(u_{j+1}\mid S_j)\cdot \frac{x-c(S_j)}{c(u_{j+1})}.
\end{equation}
Note that if $x=c(S_i)$ for some $i$, we have $j=i-1$ such that $S_{j+1}=S_i$ and hence $F(x)=f(S_i)$. For convenience, we also define $F(0)=0$.

The core relations used in the proof are the relations between the intermediate solutions obtained by the greedy heuristic and the optimal solution, especially for the intermediate greedy solutions obtained at the time when the first and second elements in the optimal solution are considered by the greedy heuristic but not adopted due to budget violation. In particular, let $o$ (resp.\ $o^\prime$) be the first (resp.\ second) element in optimal solution $\OPT$ considered by the greedy heuristic but not added to the element set $\Sg$ due to budget violation,\footnote{If $o$ does not exist, we consider $o$ as a dummy element such that $c(o)=0$ and $\frac{f(o\mid S)}{c(o)}=0$ given any $S$, and so as for $o^\prime$.} and let $Q$ (resp.\ $Q^\prime$) be the element set constructed until $o$ (resp.\ $o^\prime$) is considered by the greedy heuristic. Clearly, $Q\subseteq Q^\prime\subseteq \Sg$ and hence $f(Q)\leq f(Q^\prime)\leq f(\Sg)$ due to the monotonicity of $f$. We later derive several important and useful relations between $f(Q)$ (or $f(Q^\prime)$) and $f(\OPT)$.


\subsection{Main Proof}

We start the proof with a useful relation between $f(S)$ and $f(T)$ for any two sets $S$ and $T$ to characterize the monotone nondecreasing submodular function $f$~\cite{Nemhauser_submodular_1978}. That is, for any monotone nondecreasing submodular set function $f$, we have~\cite{Nemhauser_submodular_1978}
\begin{equation*}
	f(T)\leq f(S)+\sum_{v\in T\setminus S} f(v\mid S).
\end{equation*}
Moreover, making use of the marginal gain $f(T\mid S)$ of adding $T$ to $S$, we further have~\cite{Nemhauser_submodular_1978}
\begin{equation*}
	f(T)\leq f(S)+\sum_{i=1}^{r} f(R_i\mid S),
\end{equation*}
where $\{R_1,\dotsc,R_r\}$ is a partition of $T\setminus S$.

In the following, we provide a lower bound on the function value of the intermediate greedy solution, utilizing the monotonicity and submodularity of $f$ and the greedy rule of \MG.
\begin{lemma}\label{lemma:approx-cost}
	For any monotone nondecreasing submodular (and non-negative) set function $f$, denote $\Ss\subseteq \Sg$ as the intermediate element set constructed by the greedy heuristic after a certain number of iterations, and let $\As$ be the corresponding abandoned element set, \ie~$\As:=A_s$ where $s=\lvert \Ss\rvert$. Given any element set $T$, if $T\cap \As=\emptyset$, we have
	\begin{equation}
	f(\Ss)\geq \Big(1-\e^{-c(\Ss)/c(T)}\Big)\cdot f(T).
	\end{equation}
\end{lemma}
\begin{proof}
	The lemma directly holds when $f(\Ss)\geq f(T)$. In the following, we consider $f(\Ss)\leq f(T)$. By the monotonicity and submodularity of $f$, we have
	\begin{equation*}
		f(T)
		\leq f(S_i)+\sum_{v\in T\setminus S_i} f(v\mid S_i)
		= f(S_i)+\sum_{v\in T\setminus S_i} \Big(c(v)\cdot \frac{f(v\mid S_i)}{c(v)}\Big).
	\end{equation*}	
	According to the greedy rule, for any $i\leq s-1$ and any $v\in T\setminus S_i$, we have 
	\begin{equation*}
		\frac{f(u_{i+1}\mid S_{i})}{c(u_{i+1})}\geq \frac{f(v\mid S_i)}{c(v)},
	\end{equation*}
	since $T\cap \As=\emptyset$. 
	Thus, we can get that
	\begin{equation*}
		f(T)
		\leq f(S_i)+\frac{f(u_{i+1}\mid S_{i})}{c(u_{i+1})}\cdot \sum_{v\in T\setminus S_i} c(v)
		\leq f(S_i)+\frac{f(u_{i+1}\mid S_{i})}{c(u_{i+1})}\cdot c(T).
	\end{equation*}
	Rearranging it yields 
	\begin{equation*}
	f(T)-f(S_{i+1})\leq \Big(1-\frac{c(u_{i+1})}{c(T)}\Big)\cdot\big(f(T)-f(S_{i})\big).
	\end{equation*}
	Moreover, we know that $1-x\leq \e^{-x}$ for any $x\geq 0$, which indicates that \[1-\frac{c(u_{i+1})}{c(T)}\leq \e^{-{c(u_{i+1})}/{c(T)}}.\] Hence, observing that $f(T)-f(S_i)\geq 0$ as $f(T)\geq f(\Ss)\geq f(S_i)$, we have
	\begin{equation*}
		f(T)-f(S_{i+1})\leq \e^{-{c(u_{i+1})}/{c(T)}}\cdot \big(f(T)-f(S_{i})\big).
	\end{equation*}
	Recursively, 
	\begin{align*}
	f(T)-f(\Ss)
	&=f(T)-f(S_s)\\
	&\leq \e^{-{c(u_{s})}/{c(T)}}\cdot \big(f(T)-f(S_{s-1})\big)\\
	&\leq \e^{-{c(u_{s})}/{c(T)}}\cdot \e^{-{c(u_{s-1})}/{c(T)}}\cdot \big(f(T)-f(S_{s-2})\big)\\
	&\leq \dotsb\\
	&\leq \e^{-\sum_{i=0}^{s-1}\frac{c(u_{i+1})}{c(T)}}\cdot f(T)\\
	&=\e^{-c(\Ss)/c(T)}\cdot f(T).
	\end{align*}
	Rearranging it immediately completes the proof.
\end{proof}


Based on Lemma~\ref{lemma:approx-cost}, we can immediately derive a relation between the greedy solution $\Sg$ and the optimal solution $\OPT$ as follows.
\begin{corollary}\label{corollary:Sg-OPT-1}
	The element set $\Sg$ constructed by the greedy heuristic satisfies
	\begin{equation}
		f(\Sg)\geq \Big(1-\e^{-c(Q)/b}\Big)\cdot f(\OPT).
	\end{equation}
\end{corollary}
\begin{proof}
	According to Lemma~\ref{lemma:approx-cost}, letting $\Ss=Q$ such that $\OPT\cap \As=\emptyset$ since no element from $\OPT$ is abandoned due to budget violation before $o$ is considered, we have \[f(Q)\geq \Big(1-\e^{-c(Q)/c(\OPT)}\Big)\cdot f(\OPT).\] Since $Q\subseteq \Sg$ and $c(\OPT)\leq b$, we then complete the proof.
\end{proof}

Corollary~\ref{corollary:Sg-OPT-1} generalizes the result for a special cardinality constraint where the costs of all elements are identical. Specifically, when $c(v)=1$ for every element $v\in V$, the greedy solution $\Sg=Q$ has the same number of elements as the optimal solution $\OPT$, \ie~$c(\Sg)=c(Q)=b$, which immediately gives that $f(\Sg)\geq (1-1/\e)\cdot f(\OPT)$~\cite{Nemhauser_submodular_1978}.

\begin{lemma}\label{lemma:marginal-submodular}
	For any monotone nondecreasing submodular set function $f$, given any set $X\subseteq V$, the marginal function $f(S\mid X):=f(S\cup X)-f(X)$ upon $X$ is also a monotone nondecreasing submodular set function with respect to $S$.
\end{lemma}
\begin{proof}
	Clearly, $f(S\mid X)$ is monotone because for any $S\subseteq T$,
	\begin{equation*}
		f(S\mid X)=f(S\cup X)-f(X)\leq f(T\cup X)-f(X)=f(T\mid X).
	\end{equation*}
	Furthermore, 
	\begin{align*}
		f(S\mid X)+f(T\mid X)
		&=f(S\cup X)-f(X)+ f(T\cup X)-f(X)\\
		&\geq f((S\cup X)\cup (T\cup X))+f((S\cup X)\cap (T\cup X))-2f(X)\\
		&=f((S\cup T)\cup X)-f(X)+f((S\cap T)\cup X)-f(X)\\
		&=f(S\cup T\mid X)+ f(S\cap T\mid X),
	\end{align*}
	which shows the submodularity.
\end{proof}

Based on Lemma~\ref{lemma:marginal-submodular}, we can derive another relation between $\Sg$ and $\OPT$.
\begin{corollary}\label{corollary:Sg-OPT-2}
	Let $\OPT^\prime:=\OPT\setminus (Q\cup\{o\})$. Then,
	\begin{equation}
		f(\Sg)\geq f(Q)+\Big(1-\e^{(c(Q)+c(o^\prime)-b)/c(Q)}\Big)\cdot f(\OPT^\prime\mid Q).
	\end{equation}
\end{corollary}
\begin{proof}
	According to Lemma~\ref{lemma:marginal-submodular}, we know that $f(S\mid Q)$ is a monotone nondecreasing submodular set function with respect to $S$. Then, by Lemma~\ref{lemma:approx-cost}, letting $\Ss=Q^\prime\setminus Q$ such that $\OPT^\prime \cap \As=\emptyset$, we have
	\begin{equation}\label{eqn:f(Q')}
		f(Q^\prime)=f(Q)+f((Q^\prime\setminus Q)\mid Q)\geq f(Q)+\Big(1-\e^{-(c(Q^\prime)-c(Q))/c(\OPT^\prime)}\Big)\cdot f(\OPT^\prime\mid Q).
	\end{equation}
	In addition, we note that $c(\OPT^\prime)+c(o)\leq c(\OPT)\leq b$. Meanwhile, according to the definitions of $o$ and $o^\prime$, we know that $c(Q)+c(o)>b$ and $c(Q^\prime)+c(o^\prime)> b$. As a result,
	\begin{equation}\label{eqn:c(Q')}
		\frac{c(Q^\prime)-c(Q)}{c(\OPT^\prime)}\geq \frac{b-c(o^\prime)-c(Q)}{b-c(o)}\geq\frac{b-c(o^\prime)-c(Q)}{c(Q)}.
	\end{equation}
	Combining \eqref{eqn:f(Q')}, \eqref{eqn:c(Q')} and $f(\Sg)\geq f(Q^\prime)$ completes the proof.
\end{proof}

Both relations between $\Sg$ and $\OPT$ in Corollary~\ref{corollary:Sg-OPT-1} and Corollary~\ref{corollary:Sg-OPT-2} rely on Lemma~\ref{lemma:approx-cost} requiring $\Ss$ being a set so that the value of $c(\Ss)$ is discrete. Next, we make use of $F(\cdot)$ defined in \eqref{eqn:F(x)} to extend Lemma~\ref{lemma:approx-cost} to a continuous version where $\Ss$ is allowed to be a {partial} set so that the value of $c(\Ss)$ is continuous. 
\begin{lemma}\label{lemma:approx-cost-continuous}
	For any monotone nondecreasing submodular (and non-negative) set function $f$, given any real number $x\leq c(\Sg)$ with index $j$ satisfying $c(S_{j})< x \leq c(S_{j+1})$ and any element set $T$, if $T\cap A_{j+1}=\emptyset$, we have
	\begin{equation}
	F(x)\geq \Big(1-\e^{-x/c(T)}\Big)\cdot f(T).
	\end{equation}
\end{lemma}
\begin{proof}
	The proof is analogous to Lemma~\ref{lemma:approx-cost}. In particular, similar to Lemma~\ref{lemma:approx-cost}, we have
	\begin{equation*}
		f(T)\leq f(S_j)+\frac{f(u_{j+1}\mid S_{j})}{c(u_{j+1})}\cdot c(T).
	\end{equation*}
	Applying it to $F(x)$ gives
	\begin{equation*}
		F(x)=f(S_j)+f(u_{j+1}\mid S_j)\cdot \frac{x-c(S_j)}{c(u_{j+1})}\geq f(S_j)+\big(f(T)-f(S_j)\big)\cdot \frac{x-c(S_j)}{c(T)}.
	\end{equation*}
	Rearranging it yields
	\begin{equation*}
		f(T)-F(x)\leq \Big(1-\frac{x-c(S_j)}{c(T)}\Big)\cdot\big(f(T)-f(S_{j})\big)\leq \e^{-(x-c(S_j))/c(T)}\cdot \big(f(T)-f(S_{j})\big).
	\end{equation*}
	In addition, by Lemma~\ref{lemma:approx-cost}, we directly have
	\begin{equation}
	f(T)-f(S_{j})\leq \e^{-{c(S_j})/{c(T)}}\cdot f(T).
	\end{equation}
	Putting it together gives
	\begin{equation*}
		f(T)-F(x)\leq \e^{-{x}/{c(T)}}\cdot f(T),
	\end{equation*}
	which completes the proof by rearranging it.
\end{proof}

Based on Lemma~\ref{lemma:approx-cost-continuous}, we derive a relation between the element $v^\ast$ that maximizes $f(\{v\})$ and the optimal solution $\OPT$.
\begin{corollary}\label{corollary:v-OPT}
	The element $v^\ast$ with the largest function value satisfies 
	\begin{equation*}
		f(v^\ast)\geq \frac{f(\OPT)}{2}+\frac{1}{2(1-1/\e)}\cdot \Big(f(o\mid Q)\cdot \frac{c(Q)+c(o)+c(o^\prime)-b}{c(o)}-f(Q)\Big).
	\end{equation*}
\end{corollary}
\begin{proof}
	Define $\OPT^\diamond:=\OPT\setminus\{o,o^\prime\}$ and $x^\diamond:=c(\OPT^\diamond)$. Observing that no element from $\OPT^\diamond$ is abandoned so far, by Lemma~\ref{lemma:approx-cost-continuous}, we have
	\begin{equation*}
		F(x^\diamond)\geq \Big(1-\e^{-x^\diamond/c(\OPT^\diamond)}\Big)\cdot f(\OPT^\diamond)=(1-1/\e)\cdot f(\OPT^\diamond).
	\end{equation*}
	Meanwhile, letting $j$ be the corresponding index of $x^\diamond$ with respect to the definition of $F(\cdot)$, \ie~$c(S_j)<x^\diamond\leq c(S_{j+1})$, due to monotonicity, submodularity and the greedy rule, we have
	\begin{align*}
		F(x^\diamond)
		&\leq f(Q)+\sum_{v\in S_j\setminus Q} f(v\mid Q)+ f(u_{j+1}\mid Q)\cdot \frac{x^\diamond-c(S_j)}{c(u_{j+1})}\\
		&\leq f(Q)+\frac{f(o\mid Q)}{c(o)}\cdot \Big(\sum_{v\in S_j\setminus Q}c(v) +\big(x^\diamond-c(S_j)\big)\Big)\\
		&=f(Q)+\frac{f(o\mid Q)}{c(o)}\cdot\big(x^\diamond-c(Q)\big).
	\end{align*}
	Observing that $x^\diamond\leq b-c(o)-c(o^\prime)$, we can further get that
	\begin{equation*}
		F(x^\diamond)\leq f(Q)+\frac{f(o\mid Q)}{c(o)}\cdot\big(b-c(o)-c(o^\prime)-c(Q)\big).
	\end{equation*}
	In addition, 
	\begin{equation*}
		f(\OPT)\leq f(\OPT^\diamond)+f(o\mid\OPT^\diamond)+f(o^\prime\mid\OPT^\diamond)\leq f(\OPT^\diamond)+2f(v^\ast).
	\end{equation*}
	Putting it together, we have
	\begin{equation*}
		(1-1/\e)\cdot \big(f(\OPT)-2f(v^\ast)\big)\leq f(Q)+\frac{f(o\mid Q)}{c(o)}\cdot\big(b-c(o)-c(o^\prime)-c(Q)\big).
	\end{equation*}
	Rearranging it completes the proof.
\end{proof}

Now, we are ready to derive a lower bound on the worst-case approximation of \MG by solving the following optimization problem.
\begin{lemma}\label{lemma:ratio-alpha}
	It holds that $f(\Sm)\geq \alpha^\ast \cdot f(\OPT)$, where $\alpha^\ast $ is the minimum of the following optimization problem with respect to $\alpha,x_1,x_2,x_3,x_4,x_5,x_6$.
	\begin{align}
	\min\quad &\alpha\label{obj}\\
	\text{s.t.\quad} 
	& \alpha \geq x_1,\label{cons:lower-fs}\\
	& \alpha \geq x_1+(1-\e^{(x_4+x_6-1)/x_4})x_3,\label{cons:lower-fs-plus}\\
	& \alpha \geq x_2,\label{cons:lower-fu}\\
	& x_1 \geq (1-1/\e)(1-2\alpha)+(x_4+x_5+x_6-1)x_2/x_5,\label{cons:fs-lower-1}\\
	& x_1 \geq 1-\e^{-x_4},\label{cons:fs-lower-2}\\
	& x_1+x_2+x_3 \geq 1,\label{cons:sub1}\\
	& x_1 +x_2/x_5\geq 1,\label{cons:sub2}\\
	& x_4+x_5\geq 1,\label{cons:budget}\\
	& \alpha,x_1,x_2,x_3,x_4,x_5,x_6\in [0,1].\label{cons:domain-all}
	\end{align}
\end{lemma}

The key idea to prove Lemma~\ref{lemma:ratio-alpha} is that we consider $\alpha=\frac{f(\Sm)}{f(\OPT)}$, $x_1=\frac{f(Q)}{f(\OPT)}$, $x_2=\frac{f(o\mid Q)}{f(\OPT)}$, $x_3=\frac{ f(\OPT^\prime\mid Q)}{f(\OPT)}$, $x_4=\frac{c(Q)}{b}$, $x_5=\frac{c(o)}{b}$, $x_6=\frac{c(o^\prime)}{b}$, and show that they satisfy all the constraints of \eqref{cons:lower-fs}--\eqref{cons:domain-all} according to the relations between the solution value and the optimum derived above, and the budget constraint. This immediately gives that $f(\Sm)\geq \alpha^\ast \cdot f(\OPT)$. Interested readers are referred to Appendix~\ref{sec:appendix} for details of all the missing proofs. 

\begin{lemma}\label{lemma:alpha-bound}
	We have $\alpha^\ast\geq\alpha^\bot$, where  $\alpha^\bot$ is defined in \eqref{eqn:alpha^bot} of Theorem~\ref{thm:approx-ratio}.
\end{lemma}

\begin{proof}[Proof of Theorem~\ref{thm:approx-ratio}]
	Combining Lemmas~\ref{lemma:ratio-alpha} and \ref{lemma:alpha-bound} immediately leads to Theorem~\ref{thm:approx-ratio}.
\end{proof}

\vspace{-2mm}

\subsection{Discussion}
Our analysis procedure consists of three steps\textemdash (i) deriving relations between the objective value and the optimum, (ii) leveraging these relations to construct an optimization problem involving the approximation guarantee, and (iii) solving the optimization problem to obtain a lower bound on the approximation. For step (i), our results that utilize cost function (\eg~Lemma~\ref{lemma:approx-cost}) and continuous extension (\eg~Lemma~\ref{lemma:approx-cost-continuous}) are useful to characterize the relations between the objective value and the optimum. For step (ii), one may add more constraints to the optimization problem so that a tighter approximation factor may be obtained by step (iii). These approaches provide insights on further study of approximation analysis for submodular optimization.

	\section{Data-Dependent Upper Bound}\label{sec:problem2}
The constant approximation factor $0.405$ established above gives a lower bound on the worst-case solution quality over all problem instances. In this section, we enhance the modified greedy algorithm to derive a data-dependent upper bound on the optimum. The upper bound allows us to obtain a potentially tighter data-dependent ratio between the solution value of modified greedy and the optimum for individual problem instances.

Specifically, given a set $S$, let $v_1,v_2,\dotso$ be the sequence of elements in $V\setminus S$ in the descending order of $\frac{f(v\mid S)}{c(v)}$. Let $r$ be the lowest index such that the total cost of the elements $\{v_1,v_2,\dotsc,v_r\}$ is larger than $b$, i.e., \[c^*:=\sum_{i=1}^{r-1}c(v_i)\leq b \quad\text{and}\quad c^*+c(v_r)>b.\] We define $\Delta(b\mid S)$ as
\begin{equation}\label{eq:max_marginal_budgeted}
\Delta(b\mid S):=\sum\nolimits_{i=1}^{r-1}f(v_i\mid S)+f(v_r\mid S)\cdot\frac{b-c^*}{c(v_r)},
\end{equation}
which is an upper bound on the largest marginal gain on top of $S$ subject to the budget $b$. Specifically, let $w_i=f(v_i\mid S)$ and $c_i=c(v_i)$, then $\Delta(b\mid S)$ is the optimum of a linear program 
\begin{equation*}
	\max\sum_{i}{w_ix_i}\text{ s.t.\@ }\sum_{i}c_ix_i\leq b \text{ and }  0\leq x_i\leq 1 \text{ for every } i.
\end{equation*}
On the other hand, the largest marginal gain $\max_{c(T)\leq b}\sum_{v\in T}f(v\mid S)$ is the optimum of the corresponding integer linear program, \ie
\begin{equation*}
	\max\sum_{i}{w_ix_i}\text{ s.t.\@ }\sum_{i}c_ix_i\leq b \text{ and } x_i\in\{0,1\} \text{ for every } i.
\end{equation*}
Thus, $\Delta(b\mid S)$ is an upper bound on $\max_{c(T)\leq b}\sum_{v\in T}f(v\mid S)$. Observe that $\sum_{v\in \OPT\setminus S}f(v\mid S)$ is no more than the latter. Therefore, we have 
\begin{equation}\label{eq:upperbounds_budgeted}
	f(S)+\Delta(b\mid S)\geq f(\OPT\cup S)\geq f(\OPT).
\end{equation}
To incorporate into \MG, we choose the smallest upper bound $\Lambda$ over all the intermediate sets constructed by the greedy heuristic, \ie
\begin{equation}\label{eqn:Lambda}
\Lambda:=\min_{i}\{f(S_i)+\Delta(b\mid S_i)\},
\end{equation}
where $S_i$ contains the first $i$ elements added to $\Sg$ by the greedy heuristic. Apparently, $\Lambda$ is an upper bound of $f(\OPT)$. Algorithm~\ref{alg:UB} presents the algorithm based on \MG, which slightly modifies the algorithm by simply adding two lines for calculating the upper bound $\Lambda$ (Lines~\ref{alg:UB_initilization} and \ref{alg:UB_updatebound}). In each iteration of the greedy heuristic, it takes $O(n)$ time to find $u$ and $O(n\log n)$ to sort the elements and compute the upper bound. Thus, the above enhancement increases the time complexity of modified greedy by a multiplicative factor of $\log n$ only. Next, we show that $\Lambda$ is guaranteed to be smaller than $\frac{f(\Sm)}{0.357}$.

\begin{algorithm}[!t]
	\caption{\textsc{MGreedyUB}}\label{alg:UB}
	initialize $\Sg\gets\emptyset$, $V'\gets V$\;\label{alg:DB_filter}
	$\Lambda\leftarrow f(\Sg)+\Delta(b\mid \Sg)$\label{alg:UB_initilization}\tcp*{for bound}
	\While{$V'\neq \emptyset$}{
		find $u\gets\arg\max_{v\in V'}\big\{\frac{f(v\mid \Sg)}{c(v)}\big\}$\;\label{alg:UB_argmax}
		\If{$c(S)+c(u)\leq b$\label{alg:UB_constraint}}{
			$\Sg\gets \Sg\cup\{u\}$\;\label{alg:UB_update}
			\lIf{$\Lambda>f(\Sg)+\Delta(b\mid \Sg)$}{update the upper bound $\Lambda\leftarrow f(\Sg)+\Delta(b\mid \Sg)$\label{alg:UB_updatebound}}
		}
		update the search space $V'\leftarrow V'\setminus\{u\}$\;\label{alg:UB_reducespace}
	}
	$v^\ast\gets \arg\max_{v\in V} f(v)$\;\label{alg:UB_solution1}
	$\Sm\gets\arg\max_{S\in\{\{v^*\},\Sg\}}f(S)$\;\label{alg:UB_solution}
	\Return{$\Sm$ and $\Lambda$}\;\label{alg:UB_eturn}
\end{algorithm}

\begin{theorem}\label{theorem:approximability_budgeted}
	Let $\alpha^\prime$ be the root of \[(1-\alpha^\prime)\cdot (\ln(1-\alpha^\prime)+2)-1=0\] satisfying $\alpha^\prime>0.357$. We have 
	\begin{equation}
		f(\OPT)\leq \Lambda\leq \frac{f(\Sm)}{\alpha^\prime}\leq \frac{f(\OPT)}{\alpha^\prime}.
	\end{equation}
\end{theorem}
\begin{proof}
	The first and third inequalities are straightforward, and we show that $\Lambda\leq {f(\Sm)}/{\alpha^\prime}$. 
%
%
	The inequality is trivial if $\Lambda=f(\Sm)$. Suppose $\Lambda>f(\Sm)$. Let $S_k = \{u_1, u_2, \dotsc, u_k\}$ be the element set constructed by the greedy heuristic when the first element $u_{k+1}$ from $V^\prime$ is considered but not added to $S_k$ due to budget violation. For any $i=0,1,\dotsc, k$ and any element $v\in V^\prime$, by the greedy rule, it holds that \[\frac{f(u_{i+1}\mid S_{i})}{c(u_{i+1})}\geq \frac{f(v\mid S_{i})}{c(v)}.\] 
	Thus, 
	\begin{equation}\label{eqn:largest_marginal_k_budgeted}
		f(S_{i})+b\cdot\frac{f(u_{i+1}\mid S_{i})}{c(u_{i+1})}\geq f(S_i)+\Delta(b\mid S_i)\geq \Lambda.
	\end{equation}
	Using an analogous argument to the proof of Lemma~\ref{lemma:approx-cost}, we can get that \[f(S_k)\geq (1-\e^{-c(S_k)/b})\cdot \Lambda.\] This implies that
	\begin{equation*}
	\frac{c(S_k)}{b}\leq -\ln\Big(1-\frac{f(S_k)}{\Lambda}\Big)\leq -\ln\Big(1-\frac{f(\Sm)}{\Lambda}\Big).
	\end{equation*}
	In addition, we can directly obtain from \eqref{eqn:largest_marginal_k_budgeted} that \[f(S_k)+b\cdot\frac{f(u_{k+1}\mid S_{k})}{c(u_{k+1})}\geq \Lambda.\] This implies that
	\begin{equation*}
		\frac{c(u_{k+1})}{b}\leq \frac{f(u_{k+1}\mid S_{k})}{\Lambda-f(S_k)}\leq \frac{f(\Sm)}{\Lambda-f(\Sm)}.
	\end{equation*}
	By the algorithm definition, we know that \[c(S_k)+c(u_{k+1})>b.\] Putting it together gives
	\begin{equation*}
		\frac{f(\Sm)}{\Lambda-f(\Sm)}-\ln\Big(1-\frac{f(\Sm)}{\Lambda}\Big)\geq 1.
	\end{equation*}
	Define \[g(x):=\frac{x}{1-x}-\ln(1-x)-1\] subject to $x\in [0,1)$. One can see that $g(x)$ increases along with $x$. Thus, the minimum $x$ satisfying $g(x^\ast)\geq 0$ is achieved at $g(x^\ast)=0$ such that \[(1-x^\ast)\cdot \big(\ln(1-x^\ast)+2\big)-1=0.\] Therefore, \[\frac{f(\Sm)}{\Lambda}\geq x^\ast=\alpha^\prime.\] This completes the proof.
\end{proof}



Theorem~\ref{theorem:approximability_budgeted} shows that the data-dependent ratio of $f(\Sm)$ to $\Lambda$ is guaranteed to be larger than $0.357$ for any problem instance, which is again tighter than the factor of $\frac{1 - 1/\e}{2} \approx 0.316$ given by \citet{Khuller_BMCP_1999} and matches that given by \citet{Wolsey_MGreedy_1982}. Our proof of Theorem~\ref{theorem:approximability_budgeted} is an extension of Wolsey's analysis~\cite{Wolsey_MGreedy_1982} that generalizes the result therein. (For the unit cost version, the factor can be improved to $(1-1/\e)$ as shown in Section~\ref{subsec:ub-cardinality}.) Next, we conduct experiments to show that the data-dependent ratio is usually much larger than $0.357$ or $0.405$ in practice, which demonstrates the tightness of our upper bound $\Lambda$. \figurename~\ref{fig:bounds} depicts the relationship among $f(\Sm)$, $\Lambda$, and $f(\OPT)$.

\begin{figure}[!t]
	\centering
	\includegraphics[width=0.6\linewidth]{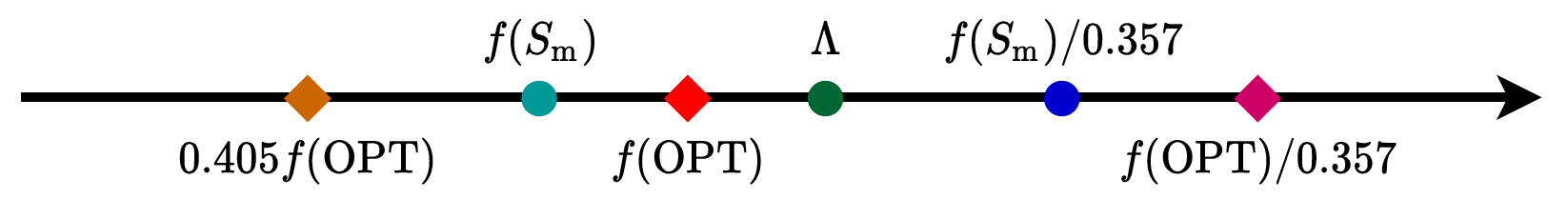}
	\vspace{-0.05in}
	\caption{Relationship among $f(\Sm)$, $\Lambda$, and $f(\OPT)$.}\label{fig:bounds}
	\vspace{-0.1in}
\end{figure}

\section{Experiments}\label{sec:experiments}
We carry out experiments on two applications to demonstrate the effectiveness of our upper bound. All the experiments are conducted on a Windows machine with an Intel Core 2.6GHz i7-7700 CPU and 32GB RAM.

\spara{Viral marketing in social networks}
%
Viral marketing in social networks \cite{Kempe_maxInfluence_2003,Tang_ASM_2019,Tang_IMhop_2018,Tang_OPIM_2018,Tang_infMax_2017,Han_AIM_2018,Huang_AIM_2020,Huang_ATPM_2020,Tang_profitMax_2018,Tang_profitMax_2016,Tang_profitMaxUS_2018} is one of the most important topics in data mining in recent years. In this application, we consider influence maximization~\cite{Kempe_maxInfluence_2003} on a social network $G=(V,E)$ with a set $V$ of vertices (representing users) and a set $E$ of edges (representing connections among users). The goal is to seed some users with incentives (e.g., discount, free samples, or monetary payment) to boost the revenue by leveraging the word-of-mouth effects on other users. We adopt the widely-used influence diffusion model called the independent cascade model \cite{Kempe_maxInfluence_2003}. Each edge $(u,v)$ is associated with a propagation probability $p_{u,v}$. Initially, the seed vertices $S$ are \textit{active}, while all the other vertices are \textit{inactive}. When a vertex $u$ first becomes active, it has a single chance to activate each inactive neighbor $v$ with success probability $p_{u,v}$. This process repeats until no more activation is possible. The \textit{influence spread} $f(S)$ of the seed set ${S}$ is the expected number of active vertices produced by the above process. \citet{Kempe_maxInfluence_2003} show that $f(S)$ is nondecreasing monotone submodular. We consider budgeted influence maximization that aims to find a vertex set $S$ maximizing $f(S)$ with the total cost $c(S)$ capped by a budget $b$, where each vertex $v$ is associated with a distinct cost $c(v)$ and $c(S)=\sum_{v\in S}c(v)$. 


Note that the influence diffusion is a random process. We use the advanced sampling technique in \cite{Ohsaka_prunedMC_2014} to estimate the influence spread in which $200$ random Monte-Carlo subgraphs are generated. We experiment with four real datasets from \cite{Kwak_twitter_2010,Leskovec_SNAP_2014} with millions of vertices, namely, 
Pokec ($1.6$M vertices and $30.6$M edges), Orkut ($3.1$M vertices and $117.2$M edges), LiveJournal ($4.8$M vertices and $69.0$M edges), and Twitter ($41.7$M vertices and $1.5$G edges).
As in \cite{Kempe_maxInfluence_2003}, we set $p_{u,v}$ of each edge $(u,v)$ to the reciprocal of $v$'s in-degree, and set $c(v)$ proportional to $v$'s out-degree to emulate that popular users require more incentives to participate. 

Due to massive data sizes, we cannot compute the true optima. To better visualize the tightness of different bounds on the optimum, we measure the ratios of the solution values obtained by \MG to the upper bounds, \eg~$f(\Sm)/\Lambda$, which represent the approximation guarantees achieved by \MG. We note that \citet{Leskovec_CELF_2007} developed an upper bound of $f(\Sg)+\Delta(b\mid \Sg)$ in our notations on the optimum. For comparisons, we evaluate both the ratios obtained for our upper bound $\Lambda$ and the upper bound developed by \citet{Leskovec_CELF_2007}. \figurename~\ref{fig:MSMK} shows the results. Note that a larger ratio represents a tighter upper bound. We observe that the ratio calculated by our upper bound is usually better than $0.9$, which is much larger than both the constant factor of $0.405$ and the ratio calculated by Leskovec~\etal's bound. This demonstrates that our upper bound $\Lambda$ is quite close to the optimum for the tested cases.


\begin{figure}[!t]
	\centering
		\vspace{-0.15in}
	\subfloat[Pokec]{\includegraphics[width=0.45\linewidth]{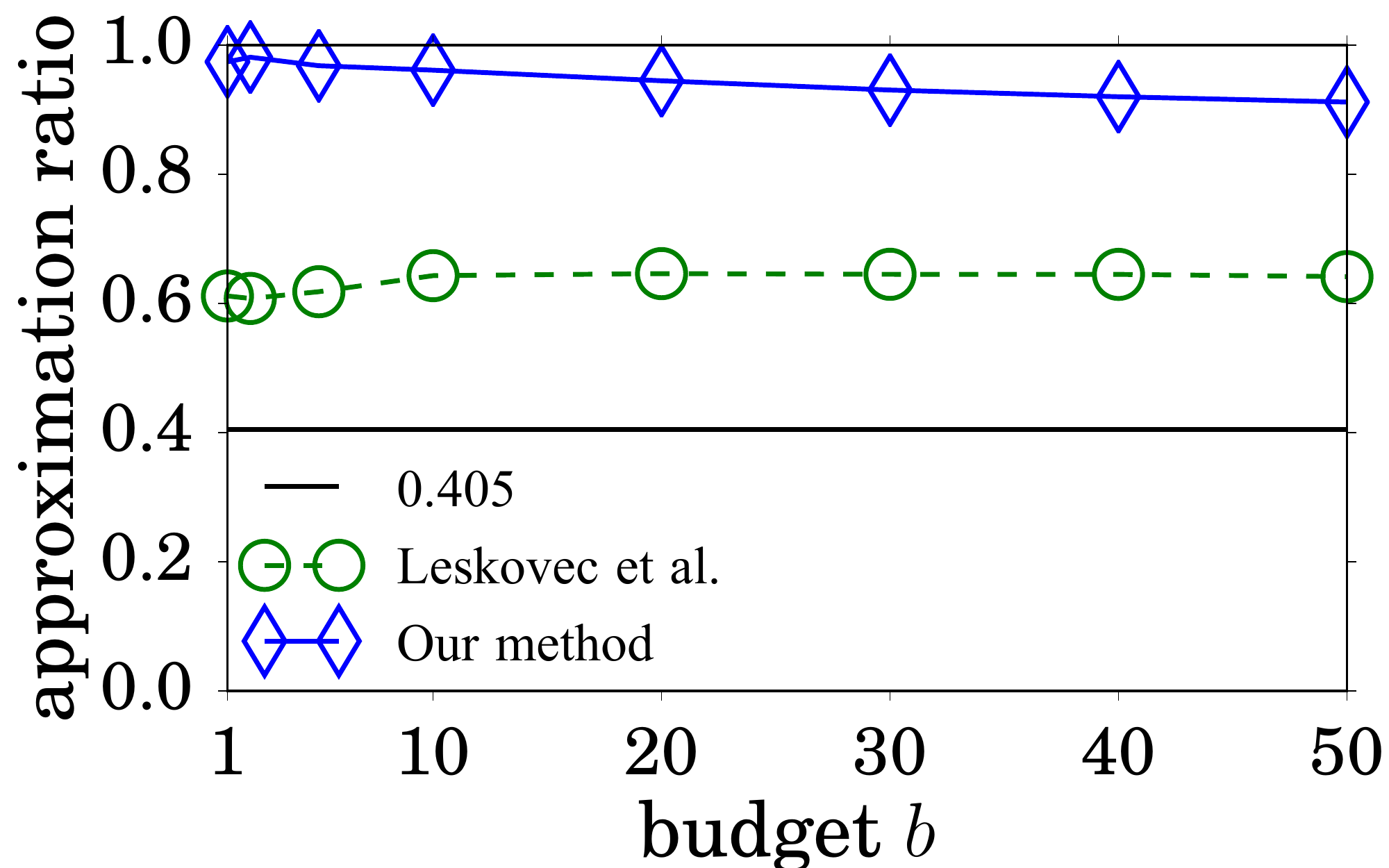}\label{subfig:soc-pokec_mgreedy}}\hfil
	\subfloat[Orkut]{\includegraphics[width=0.45\linewidth]{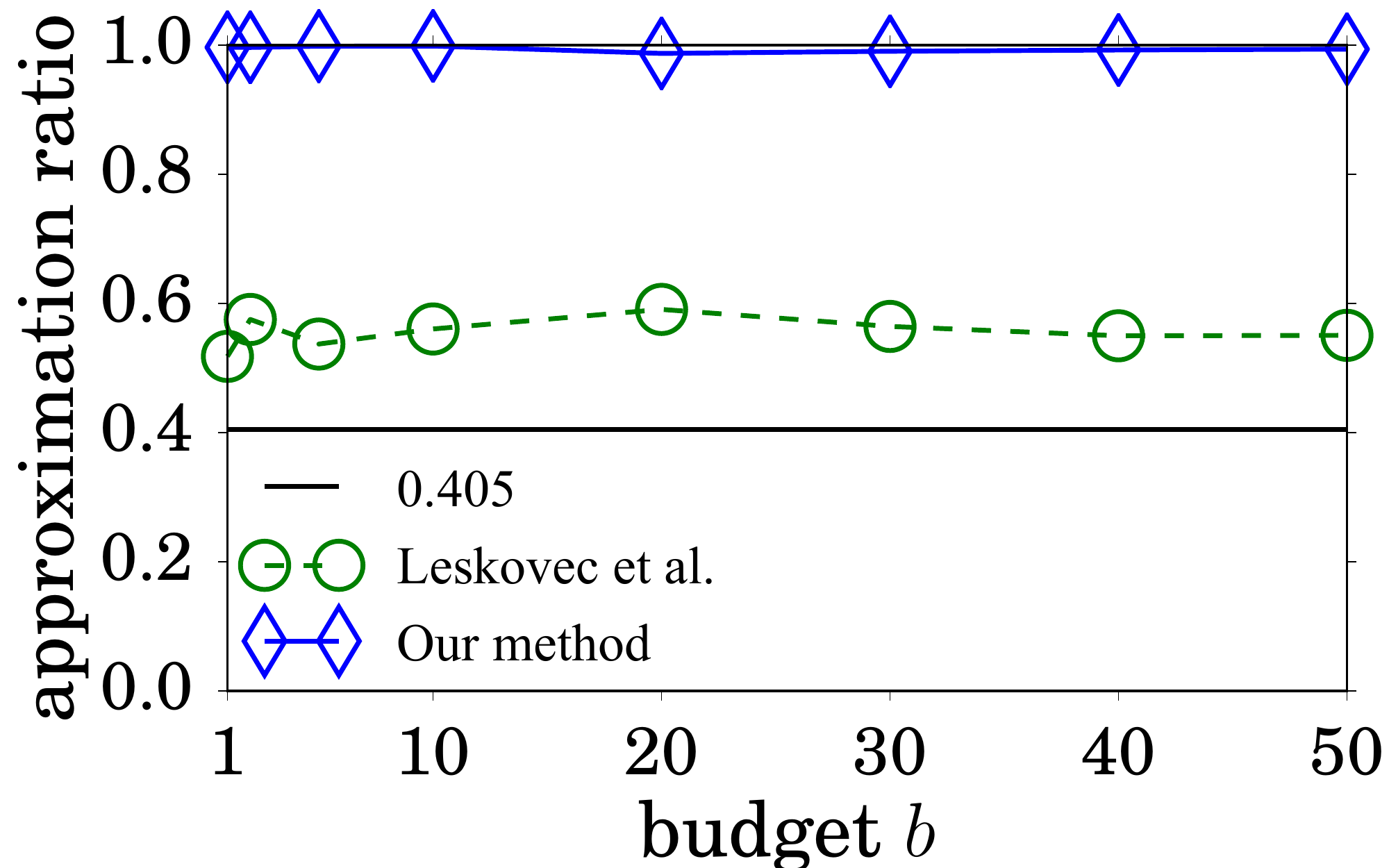}\label{subfig:orkut_mgreedy}}\hfil
	\subfloat[LiveJournal]{\includegraphics[width=0.45\linewidth]{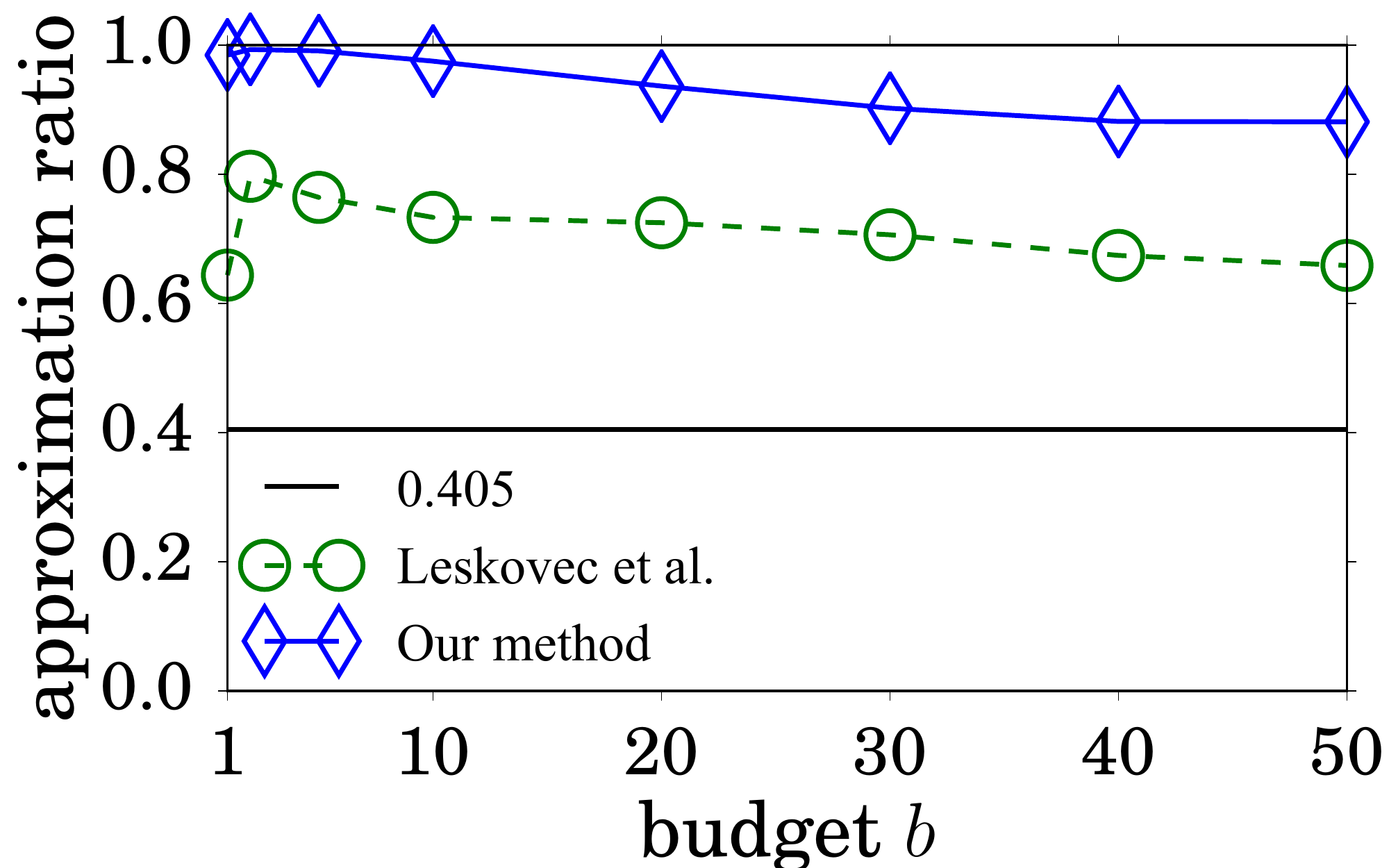}\label{subfig:liveJournal_mgreedy}}\hfil
	\subfloat[Twitter]{\includegraphics[width=0.45\linewidth]{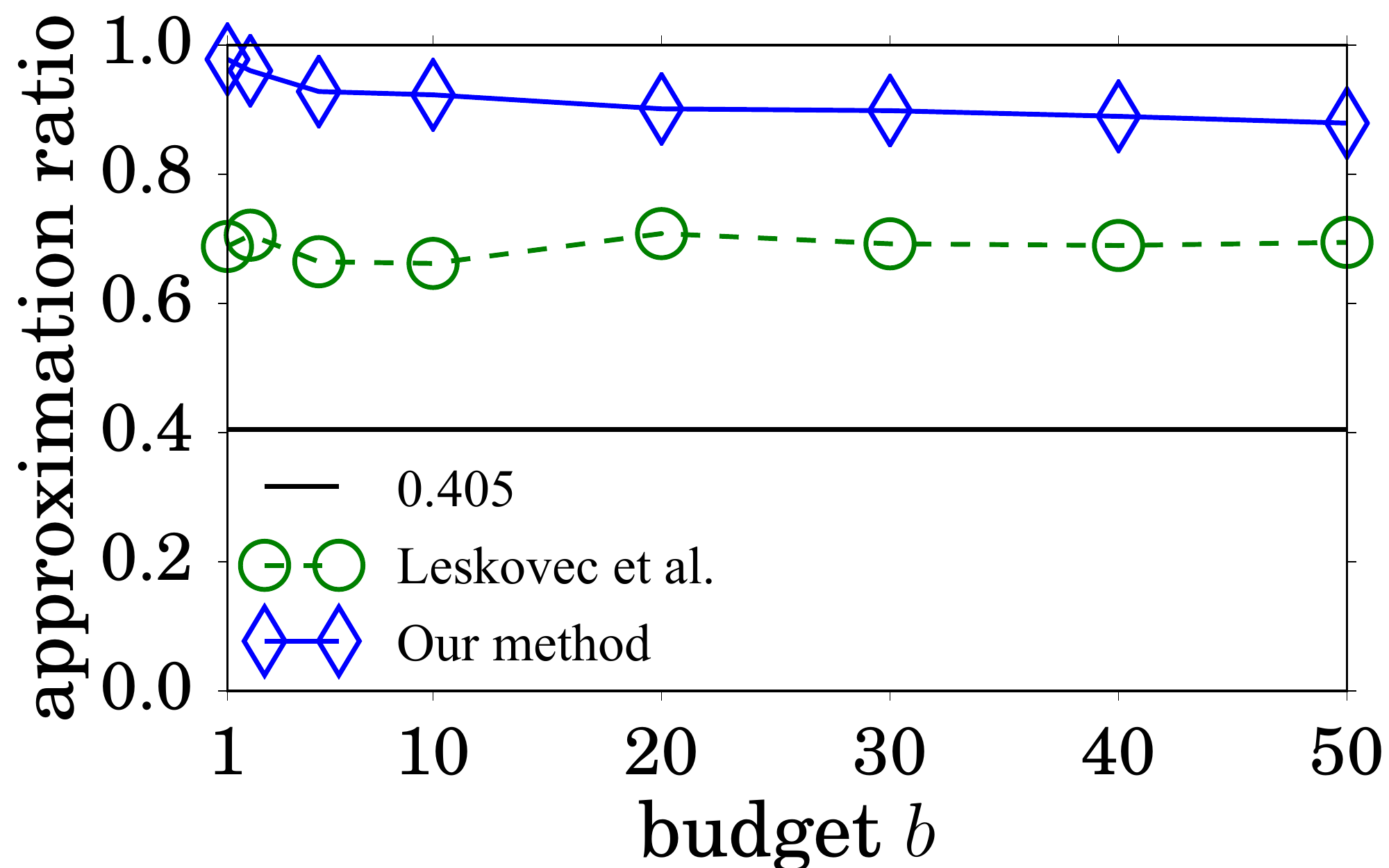}\label{subfig:big_twitter_mgreedy}}
		\vspace{-0.05in}
	\caption{Approximation ratio of modified greedy algorithm calculated by different upper bounds.}\label{fig:MSMK}
		\vspace{-0.15in}
\end{figure}




\spara{Branch-and-bound algorithm for budgeted maximum coverage} Tight bounds are valuable to advancing algorithmic efficiency. Consider the information retrieval problem where one is given a bipartite graph constructed between a set $V$ of objects (\eg~documents, images etc.) and a bag of words $W$. There is an edge $e_{v,w}$ if the object $v\in V$ contains the word $w\in W$. A natural choice of the function $f$ has the form $|\Gamma(X)|$, where $\Gamma(X)$ is the neighborhood function that maps a subset of objects $X\subseteq V$ to the set of words $\Gamma(X)\subseteq W$ presented in the objects. Meanwhile, selecting an object $v\in X$ will incur a cost $c(v)$. Intuitively, one may want to maximize the diversity (\ie~the number of words) by selecting a set of objects subject to a cost budget $b$. This problem can be seen as budgeted maximum coverage.  As a proof-of-concept, we use synthetic data that define $|V|=100$ and $|W|=100$ and randomly generate an edge between $v\in V$ and $w\in W$ with probability $p=0.02$. 
We report the average results of $10$ instances. 

We compare the branch-and-bound algorithm using our upper bound $\Lambda$ (called ``Our method'') against the data-correcting algorithm (called DCA) \cite{Goldengorin_dca_1999} which is a branch-and-bound algorithm for maximizing a submodular function.
In particular, in each branch of a search lattice $[A,B]$, the branch-and-bound algorithm needs to find an upper bound on the value of any candidate solution $S$ satisfying $A\subseteq S\subseteq B$ and $c(S)\leq b$. To achieve this goal, we first compute an upper bound $\Lambda^\prime$ on the optimum of $\max_{T\subseteq B\setminus A}\{ f(T\mid A)\colon c(T)\leq b-c(A)\}$ as $f(T\mid A)$ is also a monotone submodular function. Then, $f(A)+\Lambda^\prime$ is an upper bound on the optimum of branch $[A,B]$.
On the other hand, DCA uses $f(A)+\Delta(b-c(A)\mid A)$ as the upper bound, which is always (much) looser than ours.
DCA considers homogeneous costs only and thus we set $c(v)=1$ for each object $v$. We manually terminate the algorithm if it cannot finish within $2$ hours.
Table~\ref{table:results_MSMC} shows the running time of DCA and our algorithm when the cost budget $b$ increases from $1$ to $10$. As can be seen, our algorithm can find the optimal solution within $2$ seconds for all the cases tested whereas DCA runs 1--4 orders of magnitude slower than our algorithm when $b\leq 5$ and even fails to find the solution within $2$ hours when $b\geq 6$. 

\begin{table}[!t]
	\centering
	\caption{Running time (seconds). The field with ``\textendash'' means that the method cannot finish.}
		\vspace{-0.02in}
	\label{table:results_MSMC}
	\begin{tabular}{lrrrrrrrrrr}
		\toprule
		Budget $b$ & 1 & 2 & 3 &4 &5 &6 &7 &8 &9 &10  \\ 
		\midrule
		DCA & 0.43 & 6.06 & 99.92 & 899.66 & 6807.67 & \textendash & \textendash & \textendash & \textendash & \textendash \\ 
		\textbf{Our method} & 0.04 & 0.06 & 0.09 & 0.11 & 0.17 & 0.32 & 0.51 & 0.77 & 1.21 & 1.90 \\ 
		\bottomrule
	\end{tabular}
		\vspace{-0.05in}
\end{table}
	\section{Further Discussions}

\subsection{Upper Bound on the Approximation of \MG}
We provide an instance of the problem for which the \MG algorithm achieves a ratio of $(1/2+\varepsilon)$ for any given $\varepsilon>0$. Specifically, we consider a modular set function $f$ such that $f(S)=\sum_{v\in S} f(v)$. Suppose that there are three elements $u, v, w$ with $f(u)=f(v) = 1$, $f(w) = 1 + 2\varepsilon$, $c(u) = c(v) = 1$, and $c(w) =1+\varepsilon$. When $b=2$, \MG will select $\{w\}$ with $f(\Sm)=1+2\varepsilon$, since $\Sg=\{v^\ast\}=\{w\}$ according to the algorithm. It is easy to verify that the optimal solution is $\{u,v\}$ with $f(\OPT)=2$. As a result, \MG provides $(1/2+\varepsilon)$-approximation. We also note that \citet{Khuller_BMCP_1999} claimed that they have constructed an instance for which \MG can only achieve a ratio of approximately $0.44$, but unfortunately they did not provide the detailed instance. It is unclear whether the approximation factor of $0.405$ for \MG is completely tight, but the above analysis shows that the gap between our ratio and the actual worst-case ratio is small.

\subsection{Analysis of \texorpdfstring{$\bm{(1-1/\sqrt{\e})}$}{1-1/\unichar{"221A}e} Approximation Guarantee}\label{subsec:correct-analysis}
\citet{Khuller_BMCP_1999} claimed that the modified greedy algorithm, referred to as \MG, achieves an approximation guarantee of $(1-1/\sqrt{\e})$, but unfortunately their proof was flawed as pointed out by \citet{Zhang_billboard_2018}. 
We provide here a brief explanation of the problem in the proof of \cite[Theorem~3]{Khuller_BMCP_1999}. When showing that $f(\Sg)\geq (1-1/\sqrt{\e})\cdot f(\OPT)$ when $c(\Sg)\geq b/2$, where $\Sg$ is the set obtained by the greedy heuristic and $\OPT$ is the optimal solution, the proof relies on $c(S_\ell)\geq b/2$, where $S_\ell\subseteq \Sg$ is an intermediate set constructed by the greedy heuristic when the first element from $\OPT$ is selected for consideration but not added to $\Sg$ due to budget violation. However, there is a gap here, since $c(\Sg)\geq b/2$ does not imply $c(S_\ell)\geq b/2$. Interested readers are referred to the detailed analysis by \citet{Zhang_billboard_2018}.

In this section, we provide a correct proof for the factor of $(1-1/\sqrt{\e})$. We again utilize our general approach for analyzing approximation guarantees of algorithms by solving an optimizing problem that characterizes the relations between the solution value and the optimum. The key difference for deriving the two factors (\ie~$0.405$ in Section~\ref{sec:proof} and $(1-1/\sqrt{\e})$ in this section) lies in the relations between $f(\Sm)$ and $f(\OPT)$ used in the analysis. In particular, in Section~\ref{sec:proof}, we make use of two relatively complicated relations given by Corollary~\ref{corollary:Sg-OPT-2} and Corollary~\ref{corollary:v-OPT}, whereas in this section, we will replace them with a new and simple relation given by the following Corollary~\ref{corollary:Sg-OPT-3} which will remarkably simplify the analysis.  
\begin{theorem}\label{thm:approx-ratio-mgreedy-append}
	The \MG algorithm achieves an approximation factor of $(1-1/\sqrt{\e})$.
\end{theorem}

To prove Theorem~\ref{thm:approx-ratio-mgreedy-append}, we start with the following useful lemma.
\begin{lemma}\label{lemma:lower-bound-append}
	Given any element set $T$, the greedy heuristic returns $\Sg$ subject to a budget constraint $b$ satisfying
	\begin{equation*}
	f(\Sg)\geq \Big(1-\frac{c(T)}{b}\Big)\cdot f(T).
	\end{equation*}
\end{lemma}
\begin{proof}
	The lemma is trivial when $T\subseteq \Sg$. Suppose $T\setminus \Sg\neq \emptyset$. Let $S_\ell= \{u_1,u_2,\dotsc,u_\ell\}$ be the element set constructed by the greedy heuristic when the first element from $T$ is considered but not added to $S_\ell$ due to budget violation. Due to submodularity and the greedy rule, we have
	\begin{equation*}
	\frac{f(u_{1}\mid S_0)}{c(u_1)}\geq \frac{f(u_{2}\mid S_1)}{c(u_2)}\geq\dotsb \geq \frac{f(u_{\ell}\mid S_{\ell-1})}{c(u_{\ell})}\geq \max_{v\in T^\prime}\frac{f(v\mid S_{\ell})}{c(v)}\geq\frac{f(T^\prime\mid S_{\ell})}{c(T^\prime)},
	\end{equation*}
	where $T^\prime=T\setminus S_\ell$. Observe that \[f(T)\leq f(S_{\ell})+f(T^\prime\mid S_{\ell}).\] Meanwhile, 
	\begin{equation*}
	f(S_{\ell})=\sum_{i=1}^{\ell} f(u_{i}\mid S_{i-1})\geq \sum_{i=1}^{\ell} \Big(c(u_i)\cdot\frac{f(T^\prime\mid S_{\ell})}{c(T^\prime)}\Big)=c(S_{\ell})\cdot\frac{f(T^\prime\mid S_{\ell})}{c(T^\prime)}.
	\end{equation*}
	Note that by the algorithm definition, \[c(S_{\ell})+c(T^\prime)> b.\] Therefore,
	\begin{equation*}
	f(T)< \Big(1+\frac{c(T^\prime)}{b-c(T^\prime)}\Big)\cdot f(S_{\ell})= \frac{b}{b-c(T^\prime)}\cdot f(S_{\ell}) \leq \frac{b}{b-c(T)}\cdot f(S_{\ell}).
	\end{equation*}
	Rearranging it concludes the proof.
\end{proof}

Based on Lemma~\ref{lemma:lower-bound-append}, a relation between the greedy solution $\Sg$ and the optimal solution $\OPT$ can be derived as follows.
\begin{corollary}\label{corollary:Sg-OPT-3}
	Let $\OPT^\prime:=\OPT\setminus (Q\cup\{o\})$. Then,
	\begin{equation}
	f(\Sg)\geq f(Q)+\Big(1-\frac{c(Q)}{b-c(Q)}\Big)\cdot f(\OPT^\prime\mid Q).
	\end{equation}
\end{corollary}
\begin{proof}
	According to Lemma~\ref{lemma:marginal-submodular}, we know that $f(S\mid Q)$ is a monotone nondecreasing submodular set function with respect to $S$. Then, by Lemma~\ref{lemma:lower-bound-append}, we have
	\begin{equation*}
	f(\Sg)=f(Q)+f((\Sg\setminus Q)\mid Q)\geq f(Q)+\Big(1-\frac{c(\OPT^\prime)}{b-c(Q)}\Big)\cdot f(\OPT^\prime\mid Q).
	\end{equation*}
	Furthermore, by the algorithm definition, we know that 
	\begin{align*}
		&c(Q)+c(o)>b,\\
		\text{and}\quad&c(o)+c(\OPT^\prime)\leq c(\OPT)\leq b.
	\end{align*}
	Thus, 
	\begin{equation*}
		c(\OPT^\prime)<c(Q).
	\end{equation*}
	Putting it together completes the proof.
\end{proof}

Now, we can derive a lower bound on the worst-case approximation of \MG by solving the following optimization problem.
\begin{lemma}\label{lemma:ratio-alpha-append}
	It holds that $f(\Sm)\geq \bar{\alpha}^\ast \cdot f(\OPT)$, where $\bar{\alpha}^\ast $ is the minimum of the following optimization problem with respect to $\alpha,x_1,x_2,x_3$.
	\begin{align}
	\min\quad &\alpha\label{obj-append}\\
	\text{s.t.\quad} 
	& \alpha \geq x_1,\label{cons:lower1}\\
	& \alpha \geq 1-x_1-x_2,\label{cons:lower2}\\
	& \alpha \geq x_1+\Big(1-\frac{x_3}{1-x_3}\Big)\cdot x_2,\label{cons:lower3}\\
	& x_1 \geq 1-\e^{-x_3},\label{cons:fs}\\
	& \alpha,x_1,x_2,x_3\in [0,1].\label{cons:domain}
	\end{align}
\end{lemma}

Analogous to the proof of Lemma~\ref{lemma:ratio-alpha}, we consider $\alpha=\frac{f(\Sm)}{f(\OPT)}$, $x_1=\frac{f(Q)}{f(\OPT)}$, $x_2=\frac{ f(\OPT^\prime\mid Q)}{f(\OPT)}$, $x_3=\frac{c(Q)}{b}$ here with some simpler relations between the solution value and the optimum, \eg~\eqref{cons:lower2} and \eqref{cons:lower3} are used instead of \eqref{cons:lower-fs-plus} and \eqref{cons:fs-lower-1}.

\begin{lemma}\label{lemma:alpha-bound-append}
	$\bar{\alpha}^\ast\geq 1-1/\sqrt{\e}$. 
\end{lemma}

\begin{proof}[Proof of Theorem~\ref{thm:approx-ratio-mgreedy-append}]
	Combining Lemmas~\ref{lemma:ratio-alpha-append} and~\ref{lemma:alpha-bound-append} immediately gives Theorem~\ref{thm:approx-ratio-mgreedy-append}.
\end{proof}

\subsection{Data-Dependent Upper Bound for Cardinality Constraint}\label{subsec:ub-cardinality}
In this section, we show that when the knapsack constraint degenerates to a cardinality constraint, \ie~\[\max_{S\subseteq V} f(S) \text{ s.t.\@ }\lvert S\rvert\leq k,\] our upper bound $\Lambda$ is guaranteed to be smaller than $\frac{f(\Sg)}{1-1/\e}$, which matches the tight approximation factor of $(1-1/\e)$ \cite{Nemhauser_tight_1978}. Our analysis extends the result of $(1-1/\e)$-approximation using an alternative proof in a concise way.
\begin{theorem}\label{theorem:approximability-append}
	For monotone submodular maximization with a cardinality constraint, we have 
	\begin{equation}
		f(\OPT)\leq \Lambda\leq \frac{f(\Sg)}{1-1/\e}\leq \frac{f(\OPT)}{1-1/\e}.
	\end{equation}
\end{theorem}
\begin{proof}
	By monotonicity, submodularity and the greedy rule, we have
	\begin{equation*}
	f(S_{i})+k\cdot f(u_{i+1}\mid S_{i})\geq f(S_i)+\Delta(k\mid S_i)\geq \Lambda,
	\end{equation*}
	where $u_i$ is the $i$-the element selected by the greedy heuristic and $S_i:=\{u_1,u_2,\dotsc,u_i\}$, \eg~$S_k=\Sg$. 
	Rearranging it yields 
	\begin{equation*}
	\Lambda-f(S_{i+1})\leq \left(1-1/k\right)\cdot\big(\Lambda-f(S_{i})\big).
	\end{equation*}
	Recursively, we have
	\begin{equation*}
	\Lambda-f(S_{k})\leq \left(1-1/k\right)^k\cdot\big(\Lambda-f(S_{0})\big)=\left(1-1/k\right)^k\cdot\Lambda\leq 1/\e\cdot \Lambda.
	\end{equation*}
	Rearranging it completes the proof.
\end{proof}

In each iteration of the greedy heuristic, it takes $O(n)$ time to find $u$ and the largest $k$ marginal gains \cite{Blum_selection_1973}, where $n=|V|$ is the size of ground set. There are $k$ iterations in the greedy algorithm. Thus, the total complexity of deriving $\Lambda$ is $O(kn)$, which remains the same as that of greedy algorithm.

	\section{Related Work}
\citet{Nemhauser_submodular_1978} studied monotone submodular maximization with a cardinality constraint, and proposed a greedy heuristic that achieves an approximation factor of $(1-1/\e)$. For this problem, \citet{Nemhauser_tight_1978} showed that no polynomial algorithm can achieve an approximation factor exceeding $(1-1/\e)$. \citet{Feige_setcover_1998} further showed that even maximum coverage (which is a special submodular function) cannot be approximated in polynomial time within a ratio of $(1-1/\e+\varepsilon)$ for any given $\varepsilon>0$, unless $\mathrm{P}\!=\!\mathrm{NP}$. Leveraging the notion of {curvature}, \citet{Conforti_curvature_1984} obtained an improved upper bound $(1-\e^{-\kappa_f})/\kappa_f$, where $\kappa_f:=1-\min_{v\in V}\frac{f(v\mid V\setminus\{v\})}{f(\{v\})}\in [0,1]$ measures how much $f$ deviates from modularity. By utilizing multilinear extension \cite{Calinescu_multilinear_2011}, \citet{Sviridenko_curvature_2015} proposed a continuous greedy algorithm that can further improve the approximation ratio to $(1-\kappa_f/\e-\varepsilon)$ at the cost of increasing time complexity from $O(kn)$ to $\tilde{O}(n^4)$, where $k$ is the maximum cardinality of elements in the optimization domain. Different from these studies, we focus on the more general problem of monotone submodular maximization with a knapsack constraint, for which the greedy heuristic does not have any bounded approximation guarantee.

\citet{Wolsey_MGreedy_1982} proposed a modified greedy algorithm of $O(n^2)$, referred to as \MG, that gives a constant approximation factor of $0.357$ for the  problem of monotone submodular maximization with a knapsack constraint. \citet{Khuller_BMCP_1999} showed that \MG can achieve an approximation guarantee of $(1-1/\sqrt{\e})$ for the budgeted maximum coverage problem. This factor, after being extensively mentioned in the literature, was recently pointed out by \citet{Zhang_billboard_2018} to be problematic due to the flawed proof. Such a long-term misunderstanding on the factor of $(1-1/\sqrt{\e})$ becomes a critical issue needed to be solved urgently. In this paper, we show that the \MG algorithm can achieve an improved constant approximation ratio of $0.405$ through a careful analysis, which answers the open question of whether the worst-case approximation guarantee of \MG is better than $(1-1/\sqrt{\e})$. In addition, we also enhance the \MG algorithm to derive a data-dependent upper bound on the optimum, which slightly increases the time complexity of \MG by a multiplicative factor of $\log n$. We theoretically show that the ratio of the solution value obtained by \MG to our upper bound is always larger than $0.357$, which is again tighter than the approximation factor given by \citet{Khuller_BMCP_1999} and matches that given by \citet{Wolsey_MGreedy_1982}. We note that \citet{Leskovec_CELF_2007} developed an upper bound of $f(\Sg)+\Delta(b\mid \Sg)$ in our notations, which is always looser than ours. Unlike our upper bound with worst-case guarantees, the relationship between $f(\Sg)+\Delta(b\mid \Sg)$ and $f(\Sm)$ is unclear. As has been demonstrated in the experiments, our upper bound is significantly tighter than that developed by \citet{Leskovec_CELF_2007}. 

In addition to the modified greedy algorithm, \citet{Khuller_BMCP_1999} also gave a partial enumeration greedy heuristic that can achieve $(1-1/\e)$-approximation, which was later shown to be also applicable to the general submodular functions by \citet{Sviridenko_maxSub_2004}. Recently, \citet{Yoshida_curvature_2016} proposed a continuous greedy algorithm achieving a curvature-based approximation guarantee of $(1-\kappa_f/\e-\varepsilon)$. However, the time complexities of the partial enumeration greedy algorithm and the continuous greedy algorithm are as high as $O(n^5)$ and $\tilde{O}(n^5)$, respectively. These algorithms \cite{Sviridenko_maxSub_2004,Yoshida_curvature_2016} are hard to apply in practice. Some recent work~\cite{Badanidiyuru_MSMK_2014,Ene_MSMK_2019} proposed algorithms with $(1-1/\e-\varepsilon)$-approximation. These algorithms are again impractical due to the high dependency on $\varepsilon$, \ie~$(\nicefrac{\log n}{\varepsilon})^{O(\nicefrac{1}{\varepsilon^8})}n^2$ \cite{Badanidiyuru_MSMK_2014} and $(\nicefrac{1}{\varepsilon})^{O(\nicefrac{1}{\varepsilon^4})}n\log^2n$ \cite{Ene_MSMK_2019}, which are of theoretical interests only.

Table~\ref{table:results_MS} summarizes the results for monotone submodular maximization with a knapsack constraint. Through comparison, we find that there is still a gap between our newly derived approximation factor of $0.405$ and the best known factor of $(1-1/\e)$. However, there does not exist practical algorithms that can achieve the optimal approximation ratio of $(1-1/\e)$. In fact, devising efficient algorithms with approximation better than \MG, \eg~approaching $(1-1/\e)$ or beating $1/2$, is a challenging problem. 

\begin{table}[!t]
	\centering
	\caption{Comparison of approximations for monotone submodular maximization with a knapsack constraint.}
	\vspace{-0.03in}
	\label{table:results_MS}
	\setlength{\tabcolsep}{0.6em} 
	\begin{tabular}{cllcll}
		\toprule
		\multicolumn{3}{c}{\textbf{\MG (Low Complexity)}} & \multicolumn{3}{c}{\textbf{Other Algorithms (High Complexity)}} \\ \cmidrule(lr){1-3} \cmidrule(lr){4-6}
		{Method} & {Approximation} & {Time} & {Method} & {Approximation} & {Time} \\ 
		\midrule
		\cite{Wolsey_MGreedy_1982} & $0.357$ & $O(n^2)$ & \cite{Sviridenko_maxSub_2004} & $1-1/\e$ & $O(n^5)$ \\ 
		\cite{Khuller_BMCP_1999} & $\frac{(1-1/\e)}{2}\approx 0.316$ & $O(n^2)$ & \cite{Yoshida_curvature_2016} & $1-\kappa_f/\e-\varepsilon$ & $\tilde{O}(n^5)$ \\ 
		\textbf{Our} & $0.405$ & $O(n^2)$ & \cite{Badanidiyuru_MSMK_2014} & $1-1/\e-\varepsilon$ & $(\nicefrac{\log n}{\varepsilon})^{O(\nicefrac{1}{\varepsilon^8})}n^2$ \\ 
		\textbf{Our} & $\frac{f(\Sm)}{\Lambda}> 0.357$ & $O(n^2\log n)$ & \cite{Ene_MSMK_2019} & $1-1/\e-\varepsilon$ & $(\nicefrac{1}{\varepsilon})^{O(\nicefrac{1}{\varepsilon^4})}n\log^2n$ \\ 
		\bottomrule
	\end{tabular}
\vspace{-0.02in}
\end{table}

	\section{Conclusion}\label{sec:conclusion}
In this paper, we show that \MG can achieve an approximation factor of $0.405$ for monotone submodular maximization with a knapsack constraint. This factor not only significantly improves the known factors of $0.357$ and $(1-1/\e)/2\approx 0.316$ but also closes a critical gap on the misunderstood factor of $(1-1/\sqrt{\e})\approx 0.393$ in the literature. We also derive a data-dependent upper bound on the optimum that is guaranteed to be smaller than a multiplicative factor of $\frac{1}{0.357}$ to the solution value obtained by \MG. Empirical evaluations for the application of viral marketing in social networks show that our  bound is quite close to the optimum. It remains an open question to study whether the approximation factor of $0.405$ for \MG is completely tight.



	\begin{acks}
	This research is supported by \grantsponsor{NRF-RSS2016-004}{Singapore National Research Foundation}{} under grant~\grantnum{NRF-RSS2016-004}{NRF-RSS2016-004}, by \grantsponsor{MOE2019-T1-002-042}{Singapore Ministry of Education Academic Research Fund Tier 1}{} under grant~\grantnum{MOE2019-T1-002-042}{MOE2019-T1-002-042}, by the \grantsponsor{2018AAA0101204}{National Key R\&D Program of China}{} under Grant No.\@~\grantnum{2018AAA0101204}{2018AAA0101204}, by the \grantsponsor{61772491}{National Natural Science Foundation of China (NSFC)}{} under Grant No.\@~\grantnum{61772491}{61772491} and Grant No.\@~\grantnum{U1709217}{U1709217}, by \grantsponsor{AHY150300}{Anhui Initiative in Quantum Information Technologies}{} under Grant No.\@~\grantnum{AHY150300}{AHY150300}, and by \grantsponsor{CNS-1951952}{National Science Foundation Grant}{} \grantnum{CNS-1951952}{CNS-1951952}.
\end{acks}
	\appendix
\section{Missing Proofs}\label{sec:appendix}
\begin{proof}[Proof of Lemma~\ref{lemma:ratio-alpha}]
	To simplify the notations, define 
	\begin{align*}
	&c_Q:=\frac{c(Q)}{b},\ c_o:=\frac{c(o)}{b},\ c_{o^\prime}:=\frac{c(o^\prime)}{b},\ c^\prime:=\frac{c(\OPT^\prime)}{b},\\
	&f_Q:=\frac{f(Q)}{f(\OPT)},\ f_o:=\frac{f(o\mid Q)}{f(\OPT)},\ \hat{f}_o:=\frac{f(o\mid \emptyset)}{f(\OPT)},\ \hat{f}_{o^\prime}:=\frac{f(o^\prime\mid \emptyset)}{f(\OPT)},\ f^\prime :=\frac{ f(\OPT^\prime\mid Q)}{f(\OPT)},\\
	\text{and}\quad&\alpha_{\mathrm{m}}:=\frac{f(\Sm)}{f(\OPT)}.
	\end{align*}
	We show that $\alpha=\alpha_{\mathrm{m}}$, $x_1=f_Q$, $x_2=f_o$, $x_3=f^\prime$, $x_4=c_Q$, $x_5=c_o$, and $x_6=c_{o^\prime}$ are always feasible to the optimization problem defined in the lemma, which indicates that $f(\Sm)\geq \alpha^\ast \cdot f(\OPT)$.
	
	By the algorithm definition, 
	\begin{align}
	&\alpha_{\mathrm{m}}\geq \frac{f(\Sg)}{f(\OPT)}\geq f_Q, \tag{Constraint~\eqref{cons:lower-fs}}\\
	\text{and}\quad&\alpha_{\mathrm{m}}\geq \frac{f(v^\ast)}{f(\OPT)}\geq f_o. \tag{Constraint~\eqref{cons:lower-fu}}
	\end{align}
	By Corollary~\ref{corollary:Sg-OPT-2}, we have
	\begin{equation}
	\alpha_{\mathrm{m}}\geq f_Q+(1-\e^{(c_Q+c_{o^\prime}-1)/c_Q})f^\prime.\tag{Constraint~\eqref{cons:lower-fs-plus}}
	\end{equation}
	By Corollary~\ref{corollary:v-OPT}, we have 
	\begin{equation}
	f_Q\geq (1-1/\e)(1-2\alpha_{\mathrm{m}})+(c_Q+c_o+c_{o^\prime}-1)f_o/c_o.\tag{Constraint~\eqref{cons:fs-lower-1}}
	\end{equation}
	By Corollary~\ref{corollary:Sg-OPT-1}, we have 
	\begin{equation}
	f_Q\geq 1-\e^{-c_Q}. \tag{Constraint~\eqref{cons:fs-lower-2}}
	\end{equation}
	Due to monotonicity, submodularity and the greedy rule, it is easy to get that 
	\begin{align}
	&f_Q+f_o+f^\prime\geq 1,\tag{Constraint~\eqref{cons:sub1}}\\
	\text{and}\quad&f_Q+\frac{f_o}{c_o}\geq 1.\tag{Constraint~\eqref{cons:sub2}}
	\end{align}
	Meanwhile, due to budget violation 
	\begin{equation}
	c_S+c_u>1.\tag{Constraint~\eqref{cons:budget}}
	\end{equation}
	Finally, by definition, 
	\begin{equation}
	\alpha_{\mathrm{m}},f_Q,f_o,f^\prime,c_Q,c_o,c_{o^\prime}\in [0,1].\tag{Constraint~\eqref{cons:domain-all}}
	\end{equation}
	As can be seen, all constraints are satisfied, and hence the proof is done.
\end{proof}

\begin{proof}[Proof of Lemma~\ref{lemma:alpha-bound}]
	We first consider the case $x_4+x_6-1\geq 0$. We obtain from \eqref{cons:fs-lower-1} and \eqref{cons:fs-lower-2} that
	\begin{align}
	&x_1 \geq (1-1/\e)(1-2\alpha)+x_2,\label{cons:fs-lower-1-simple}\\
	\text{and}\quad&-\ln(1-x_1)\geq x_4.\label{cons:fs-lower-2-simple}
	\end{align}
	Then, $x_5\times \eqref{cons:sub2}+\eqref{cons:fs-lower-1-simple}+(1-x_1)\times(\eqref{cons:budget}+\eqref{cons:fs-lower-2-simple})$ gives
	\begin{equation*}
	x_1-(1-x_1)\ln(1-x_1)\geq 1-x_1+(1-1/\e)(1-2\alpha).
	\end{equation*}
	Rearranging yields
	\begin{equation}
	(1-x_1)(\ln(1-x_1)+2)-1/\e-2(1-1/\e)\alpha\leq 0.
	\end{equation}
	When $\alpha\geq 1-\e^{-3}$, we directly have $\alpha\geq \alpha^\bot$. When $\alpha\leq 1-\e^{-3}$, we know that $x_1\leq 1-\e^{-3}$ by \eqref{cons:lower-fs}. Then, $(1-x_1)(\ln(1-x_1)+2)$ decreases along with $x_1$, which indicates that \[(1-\alpha)(\ln(1-\alpha)+2)-1/\e-2(1-1/\e)\alpha\leq 0.\] Note that the left hand side is equal to $(1-\alpha)\ln(1-\alpha)+(2-1/\e)(1-2\alpha)$, which strictly decreases along with $\alpha$ when $\alpha\leq 1-\e^{-3}$. This implies that $\alpha\geq \alpha^\bot$.
	
	Next, we consider the case $x_4+x_6-1\leq 0$. We prove $\alpha\geq \alpha^\bot$ by contradiction. Assume on the contrary that $\alpha<\alpha^\bot$. Then, $x_1<\alpha^\bot$ and $x_2<\alpha^\bot$. We can get from \eqref{cons:lower-fs-plus} and \eqref{cons:sub1} that
	\begin{equation*}
	x_4+x_6-1\geq \ln\Big(1-\frac{\alpha-x_1}{x_3}\Big)\cdot x_4\geq \ln\Big(\frac{1-\alpha^\bot-x_2}{1-x_1-x_2}\Big)\cdot x_4.
	\end{equation*}
	Combining it with \eqref{cons:fs-lower-1} and \eqref{cons:budget} gives
	\begin{align}
	x_1
	&\geq (1-{1}/{\e})(1-2\alpha)+x_2\Big(1+\frac{x_4+x_6-1}{1-x_4}\Big)\nonumber\\
	&> (1-{1}/{\e})(1-2\alpha^\bot)+x_2\Big(1+\ln\Big(\frac{1-\alpha^\bot-x_2}{1-x_1-x_2}\Big)\cdot \frac{x_4}{1-x_4}\Big).\label{eqn:fs-lower-1}
	\end{align}
	Note that as $\ln(\tfrac{1-\alpha^\bot-x_2}{1-x_1-x_2})\leq 0$, the left hand side of the above inequality decreases along with $x_4$ when $x_4\in [0,1)$. By \eqref{cons:lower-fs} and \eqref{cons:fs-lower-2-simple}, we have \[x_4\leq -\ln(1-\alpha^\bot).\] Combining with \eqref{eqn:fs-lower-1} gives
	\begin{equation}\label{eqn:fs-lower-2}
	x_1> (1-1/\e)(1-2\alpha^\bot)+x_2\Big(1+\ln\Big(\frac{1-\alpha^\bot-x_2}{1-x_1-x_2}\Big)\cdot \frac{-\ln(1-\alpha^\bot)}{1+\ln(1-\alpha^\bot)}\Big).
	\end{equation}
	Meanwhile, $x_5\times \eqref{cons:sub2}+(1-x_1)\times(\eqref{cons:budget}+\eqref{cons:fs-lower-2-simple})$ gives
	\begin{equation}\label{eqn:fs-lower-3}
	x_2\geq (1-x_1)(1+\ln(1-x_1)),
	\end{equation}
	where the right hand side decreases along with $x_1$ when $x_1\in[0,\alpha^\bot)$ as required by \eqref{cons:lower-fs}. Then, 
	\begin{equation*}
	x_2\geq (1-\alpha^\bot)(1+\ln(1-\alpha^\bot)).
	\end{equation*}
	Define \[x^\bot_2:=(1-\alpha^\bot)(1+\ln(1-\alpha^\bot)).\] Furthermore, for the right hand side of \eqref{eqn:fs-lower-2}, define \[g(x_2):=x_2\Big(1+\ln\Big(\frac{1-\alpha^\bot-x_2}{1-x_1-x_2}\Big)\cdot \frac{-\ln(1-\alpha^\bot)}{1+\ln(1-\alpha^\bot)}\Big)\] subject to $x_2\in[x^\bot_2,\alpha^\bot]$. Taking the derivative of $g(x_2)$ with respective to $x_2$ gives
	\begin{equation*}
	g^\prime(x_2)=1+\Big(\ln\Big(\frac{1-\alpha^\bot-x_2}{1-x_1-x_2}\Big)+\frac{x_2(x_1-\alpha^\bot)}{(1-x_1-x_2)(1-\alpha^\bot-x_2)}\Big)\cdot \frac{-\ln(1-\alpha^\bot)}{1+\ln(1-\alpha^\bot)}.
	\end{equation*}
	Observe that $g^\prime(x_2)$ decreases along with $x_2$. Thus, $g(x_2)\geq \min\{g(x^\bot_2),g(\alpha^\bot)\}$. Define
	\begin{equation*}
	\tilde{g}(x_1):=g(\alpha^\bot)-g(x^\bot_2).
	\end{equation*}
	Taking the derivative of $\tilde{g}(x_1)$ with respect to $x_1$ gives
	\begin{align*}
	\tilde{g}^\prime(x_1)
	&=\Big(\frac{\alpha^\bot}{1-x_1-\alpha^\bot}-\frac{x^\bot_2}{1-x_1-x^\bot_2}\Big)\cdot\frac{-\ln(1-\alpha^\bot)}{1+\ln(1-\alpha^\bot)}\\
	&=\frac{(\alpha^\bot-x^\bot_2)(1-x_1)}{(1-x_1-\alpha^\bot)(1-x_1-x^\bot_2)}\cdot\frac{-\ln(1-\alpha^\bot)}{1+\ln(1-\alpha^\bot)}\\
	&\geq 0.
	\end{align*}
	Meanwhile, by \eqref{eqn:fs-lower-3}, we can get that \[\alpha^\bot>(1-x_1)(1+\ln(1-x_1)),\] which indicates that $x_1>0.32$. Hence, $\tilde{g}(x_1)\geq\tilde{g}(0.32)$. One can verify that
	\begin{equation*}
	\tilde{g}(0.32)
	=\alpha^\bot\Big(1+\ln\Big(\frac{1-2\alpha^\bot}{0.68-\alpha^\bot}\Big)\cdot \frac{-\ln(1-\alpha^\bot)}{1+\ln(1-\alpha^\bot)}\Big)-x^\bot_2\Big(1+\ln\Big(\frac{1-\alpha^\bot-x^\bot_2}{0.68-x^\bot_2}\Big)\cdot \frac{-\ln(1-\alpha^\bot)}{1+\ln(1-\alpha^\bot)}\Big)>0.
	\end{equation*}
	This implies that $g(x_2)\geq g(x^\bot_2)$. Therefore, \eqref{eqn:fs-lower-2} can be further relaxed to 
	\begin{equation}\label{eqn:fs-lower-4}
	x_1> (1-1/\e)(1-2\alpha^\bot)+x^\bot_2\Big(1+\ln\Big(\frac{1-\alpha^\bot-x^\bot_2}{1-x_1-x^\bot_2}\Big)\cdot \frac{-\ln(1-\alpha^\bot)}{1+\ln(1-\alpha^\bot)}\Big).
	\end{equation}
	Furthermore, define \[\hat{g}(x_1):=x^\bot_2\Big(1+\ln\Big(\frac{1-\alpha^\bot-x^\bot_2}{1-x_1-x^\bot_2}\Big)\cdot \frac{-\ln(1-\alpha^\bot)}{1+\ln(1-\alpha^\bot)}\Big)-x_1\] subject to $x_1\in[0,\alpha^\bot)$. Taking the derivative of $\hat{g}(x_1)$ with respect to $x_1$ gives
	\begin{equation*}
	\hat{g}^\prime(x_1)=\frac{x^\bot_2}{1-x_1-x^\bot_2}\cdot \frac{-\ln(1-\alpha^\bot)}{1+\ln(1-\alpha^\bot)}-1\leq \frac{x^\bot_2}{1-\alpha^\bot-x^\bot_2}\cdot \frac{-\ln(1-\alpha^\bot)}{1+\ln(1-\alpha^\bot)}-1=0.
	\end{equation*}
	This indicates that \[\hat{g}(x_1)\geq \hat{g}(\alpha^\bot)=x^\bot_2-\alpha^\bot.\] Putting it together with \eqref{eqn:fs-lower-4} gives
	\begin{equation*}
	0>(1-1/\e)(1-2\alpha^\bot)+x^\bot_2-\alpha^\bot=(1-\alpha^\bot)\ln(1-\alpha^\bot)+(2-1/\e)(1-2\alpha^\bot)=0.
	\end{equation*}
	This shows a contradiction and completes the proof.
\end{proof}

\begin{proof}[Proof of Lemma~\ref{lemma:ratio-alpha-append}]
	Again, to simplify the notations, define 
	\begin{align*}
	&c_Q:=\frac{c(Q)}{b},\ c_o:=\frac{c(o)}{b},\ c^\prime:=\frac{c(\OPT^\prime)}{b},\\
	&f_Q:=\frac{f(Q)}{f(\OPT)},\ f_o:=\frac{f(o\mid Q)}{f(\OPT)},\ f^\prime :=\frac{ f(\OPT^\prime\mid Q)}{f(\OPT)},\ \text{and}\ \alpha_{\mathrm{m}}:=\frac{f(\Sm)}{f(\OPT)}.
	\end{align*}
	In what follows, we show that $\alpha=\alpha_{\mathrm{m}}$, $x_1=f_Q$, $x_2=f^\prime$, and $x_3=c_Q$ are always feasible to the optimization problem defined in the lemma, which indicates that $f(\Sm)\geq \bar{\alpha}^\ast \cdot f(\OPT)$.
	
	By the algorithm definition, 
	\begin{equation}
	\alpha_{\mathrm{m}}\geq f_Q.\tag{Constraint~\eqref{cons:lower1}}
	\end{equation}
	Due to monotonicity, submodularity and the greedy rule, it is easy to get that 
	\begin{equation*}
	f_Q+f_o+f^\prime\geq 1.
	\end{equation*}
	Together with $\alpha_{\mathrm{m}}\geq f_o$, we have
	\begin{equation}
	\alpha_{\mathrm{m}}\geq 1-f_Q-f^\prime.\tag{Constraint~\eqref{cons:lower2}}
	\end{equation}
	By Corollary~\ref{corollary:Sg-OPT-3}, we have
	\begin{equation}
	\alpha_{\mathrm{m}}\geq f_Q+\Big(1-\frac{c_Q}{1-c_Q}\Big)\cdot f^\prime.\tag{Constraint~\eqref{cons:lower3}}
	\end{equation}
	By Corollary~\ref{corollary:Sg-OPT-1}, we have 
	\begin{equation}
	f_Q\geq 1-\e^{-c_Q}.\tag{Constraint~\eqref{cons:fs}}
	\end{equation}
	Finally, by definition,
	\begin{equation}
	\alpha_{\mathrm{m}}, f_Q, f^\prime, c_Q\in [0,1].\tag{Constraint~\eqref{cons:domain}}
	\end{equation}
	As can be seen, all constraints are satisfied, and hence the lemma is proven.
\end{proof}

\begin{proof}[Proof of Lemma~\ref{lemma:alpha-bound-append}]
	If $x_3>0.5$, constraints \eqref{cons:lower1} and \eqref{cons:fs} directly show that $\alpha \geq 1-1/\sqrt{\e}$. Next, we consider the case $0\leq x_3\leq 0.5$. Then, $(1-2x_3)\times \eqref{cons:lower2}+(1-x_3)\times \eqref{cons:lower3}+x_3\times \eqref{cons:fs}$ gives 
	\begin{equation*}
	(2-3x_3)\alpha \geq 1-x_3-x_3\e^{-x_3}.
	\end{equation*}
	Define \[g(x):=\frac{1-x-xe^{-x}}{2-3x}\] subject to $0\leq x\leq 0.5$. Taking the derivative of $g(x)$ with respect to $x$ gives \[g^\prime(x)=\frac{1-(3x^2-2x+2)\e^{-x}}{(2-3x)^2}.\] Furthermore, the derivative of $(2-3x)^2g^\prime(x)$ with respect to $x$ is $(3x^2-8x+4)\e^{-x}$, which is non-negative when $0\leq x\leq 0.5$. Thus, $(2-3x)^2g^\prime(x)$ achieves its maximum at $x=0.5$, i.e.,~the maximum is $1-1.75\e^{-0.5}<0$. This implies that $g(x)$ achieves its minimum at $x=0.5$, i.e.,~the minimum is $1-1/\sqrt{\e}$. Therefore, when $0\leq x_3\leq 0.5$, it also holds that \[\alpha\geq 1-1/\sqrt{\e},\] which concludes the lemma.
\end{proof}
	
\end{sloppy}

\bibliographystyle{ACM-Reference-Format}
\bibliography{reference}
	
\end{document}